\documentclass[12pt, reqno]{amsart}
\usepackage{amssymb}
\usepackage{amsmath}
\usepackage{graphicx,color}
\usepackage{wrapfig,framed}
\usepackage{mathtools}

\usepackage[height=22.7cm, width=16.5cm, hmarginratio={1:1}]{geometry}
\usepackage{hyperref}

\theoremstyle{plain}
\newtheorem{theorem}{Theorem}
\newtheorem{prop}{Proposition}
\newtheorem{lemma}{Lemma}
\newtheorem{cor}{Corollary}
\theoremstyle{definition}

\newcommand{\beq}{\begin{equation}}
\newcommand{\eeq}{\end{equation}}

\newcommand{\nn}{\nonumber}
\newcommand{\e}{\epsilon}
\newcommand{\p}{\partial}

\newcommand{\tr}{{\rm tr}}
\newcommand{\bs}{{\bf s}}

\newcommand{\QQ}{\mathbb{Q}}
\newcommand{\ZZ}{\mathbb{Z}}
\newcommand{\F}{\mathcal{F}}

\usepackage{array}
\usepackage{makecell}
\newcolumntype{M}[1]{>{\centering\arraybackslash}m{#1}}
\newcolumntype{R}[1]{>{\raggedleft\arraybackslash}m{#1}}
\newcolumntype{N}{@{}m{0pt}@{}}

\begin{document}

\title[On $b$-angulations of surfaces]
{On enumeration of $b$-angulations of surfaces from an integrability perspective}
\author{Elba Garcia-Failde, Jianghao Xu, Di Yang, Don Zagier}
\address{Elba Garcia-Failde, Institut de Math\'ematiques de Jussieu-Paris Rive Gauche, Sorbonne Universit\'e, 75252 Paris, France, and Departament de Matem\`{a}tiques, Universitat Polit\`{e}cnica de Catalunya, 08028 Barcelona, Spain}
\email{elba.garcia@upc.edu}
\address{Jianghao Xu, School of Mathematical Sciences, University of Science and Technology of China, 230026 Hefei, P.R.~China}
\email{xjh\_020403@mail.ustc.edu.cn}
\address{Di Yang, School of Mathematical Sciences, University of Science and Technology of China, 230026 Hefei, P.R.~China}
\email{diyang@ustc.edu.cn}
\address{Don Zagier, Max Planck Institute for Mathematics, 53111 Bonn, 
Germany, and International Centre for Theoretical Physics, 34014 Trieste, Italy}
\email{dbz@mpim-bonn.mpg.de}
\date{}
\begin{abstract}
In this paper, we study generating series enumerating polygonal angulations of closed oriented surfaces of fixed genus, focusing on $b$-angulations with $b = 3$ or $b = 2\nu$, $\nu \geq 2$. 
Based on Toda integrability, we establish new structural results in the cases $b = 3$ and $b = 4$. 
Furthermore, via the Hodge--GUE correspondence, we derive a fine structure in the $b = 2\nu$ case, which implies a conjectural statement of Gharakhloo--Latimer.
\end{abstract}
\maketitle
\tableofcontents
\section{Introduction and statements of the main results}\label{introduction}
Enumerating ribbon graphs, also known as (combinatorial) maps,
is a fundamental problem at the interface of mathematics and 
mathematical physics, attracting interest from combinatorics (cf.~\cite{BIZ,CMS,Eynard,HZ,T}), 
and revealing deep relations to quantum field theory~\cite{BIZ,BIPZ,BK,DS,GM, W} and  
geometry~\cite{BCGE,Du96, Du2, DLYZ20, HZ, Y2}.
Let $\mathcal{R}^{\rm conn}_{g}(b_1,\dots,b_k)$ be the set of connected oriented labelled\footnote{A labelled ribbon graph refers to a ribbon graph whose half-edges are labelled.} ribbon graphs of genus $g$ with $k$ vertices of valencies $b_1,\dots,b_k$, and let
$n_{g}(b_1,\dots,b_k)\coloneqq|\mathcal{R}^{\rm conn}_{g}(b_1,\dots,b_k)|$. 
Here, $g\geq 0$ and $b_1,\dots,b_k\geq 1$. 
By the Euler formula, the number $n_g(b_1,\dots,b_k)$ vanishes unless $2-2g-k+\frac{|b|}{2}$ is a positive integer. By looking at the dual graphs, one can also understand $n_g(b_1,\dots,b_k)$ as the number of polygon-angulations with $k$ polygons of sizes $b_1,\dots,b_k$ on a genus $g$ closed oriented surface.

Following~\cite{Du2,DuY},
define a power series of infinitely many variables ${\bf s}=(s_1,s_2,\dots)$ by 
\begin{align}
\F(x,{\bf s};\e) = & \, \frac{x^2}{2\e^2}\Bigl(\log x-\frac32\Bigr) 
- \frac{\log x}{12} + \zeta'(-1) + \sum_{g\geq2} \frac{\e^{2g-2} B_{2g}}{4g(g-1)x^{2g-2}}\nn\\
&+\sum_{g\geq 0}\e^{2g-2}\sum_{k\geq 1}\frac1{k!}\sum_{b_1,\dots,b_k\ge1}n_{g}(b_1,\dots,b_k) \,
s_{b_1} \cdots s_{b_k} x^{2-2g - k + \frac{|{\bf b}|}2}\,,
\label{Fgue1x'}
\end{align}
 called the {\it free energy}, 
and define $Z(x,\bs;\e)\coloneqq e^{\F(x,\bs;\e)}$, called the {\it partition function}. 
Here, $x$ is a formal variable,
$\zeta(s)$ denotes the Riemann zeta function, and $B_m$ denotes the $m$th Bernoulli number. 
It is known from e.g.~\cite{AvM,BIZ,HZ,DuY,Y2} that $Z(x,\bs;\e)$ can be understood as the following integral
\beq\label{matrixintegral}
2^{-n}\pi^{-\frac{n(n+1)}{2}} \e^{-\frac1{12}} G(n+1) \int_{{\mathcal H}(n)}
\exp\Bigl(-\frac1{\e}\tr\Bigl(\frac12 M^2-\sum_{b\geq 1}s_b M^b \Bigr)\Bigr) dM\,,
\eeq 
where $x=n\e$, $G$ denotes
Barnes's $G$-function, and  
\beq
dM = \prod_{1\leq i\leq n} d M_{ii} \prod_{1\leq i<j\leq n} d{\rm Re} M_{ij}\, d{\rm Im}M_{ij}\,.
\eeq

Define $\F_g(x,\bs)\coloneqq[\e^{2g-2}]\F(x,\bs;\e)$, $g\geq 0$, the {\it genus~$g$ free energy}.
Structures for $\F_g(x,\bs)$ have been studied in~\cite{ACKM, Du2, EO09, DLYZ20, DuY, DuY2, Y2}. By definition, 
$\F_g(x,\bs)$ encodes the enumeration of arbitrary tilings on a genus~$g$ surface
(throughout this paper, all surfaces are assumed to be closed and oriented). 
The special case of enumeration of $b$-angulations, with a fixed value of $b$, has attracted a lot of 
interest~\cite{DuY, E, E2, ELT, GL} and will be the main focus of this paper. This corresponds to restricting
$\F(x,\bs;\e)$ to $\bs=(0,\dots,0,s_b=s,0,\dots)$, and by rescaling we can assume $s=1$, so we will consider 
\begin{align}\label{Fgexplicit}
\F_g^{\{b\}}(x)\coloneqq&\,\F_g(x,{\bf s}={\bf 1}_b)
= \delta_{g,0}\frac{x^2}{2}\Bigl(\log x-\frac32\Bigr) 
+\delta_{g,1}\Bigl(-\frac{\log x}{12} + \zeta'(-1)\Bigr)\nn\\
&\,+\, \delta_{g\geq 2}\frac{B_{2g}}{4g(g-1)x^{2g-2}}
+\sum_{k\geq 1}\frac{n_{g}(b^k)}{k!}x^{2-2g+(\frac{b}2-1)k}\,,
\end{align}
where ${\bf 1}_b$ denotes the infinite vector with the $b$th component equal to 1 and all other components equal to 0.

According to e.g.~\cite{AvM, DuY, GMMMO, Y2}, $Z(x,\bs;\e)$ is a tau-function of the Toda lattice hierarchy. The latter will be one of the main tools for this paper. 
Following~\cite{Du2,DuY,DuZ,Y2}, define 
\beq
u=\frac{\p^2 \F^{\{b\}}_0(x)}{\p x^2}\,,\quad w=e^{u}\,.
\eeq

\smallskip

{\bf The case $b=3$.}
Following~\cite{BD,DuY}, 
introduce a power series $v=v(x)=6 x+324 x^2+31104 x^3+\cdots$ as the unique solution to the following cubic equation 
\beq\label{triveqn}
6 \, x = (1-9v+18v^2) \, v\,.
\eeq
According to~\cite{DuY} we have the identity
\beq\label{triweqn}
w = \frac{x}{1-6v}\,.
\eeq
(We will give a new proof of~\eqref{triweqn} in Section~\ref{triangulations}.)
It follows from~\eqref{triveqn},~\eqref{triweqn} that $w=x+36x^2+3240x^3+\cdots$ is the unique solution to the cubic equation
\beq\label{triweqn2}
x^2=w^2-72w^3\,.
\eeq
Equation~\eqref{triweqn2} was also obtained in~\cite{BD}.
Building on, for example,\cite{AvM, DuY, GMMMO} (see also~\cite{BD}) and using a method from~\cite{DYZ}, we will prove
in Section~\ref{triangulations} the following theorem.
\begin{theorem}\label{trithm1}
For $g=0$, we have $\frac{\p^2 \F^{\{3\}}_0(x)}{\p x^2}=\log w$.
For $g=1$, 
\beq\label{triF1form1333}
\mathcal{F}^{\{3\}}_1(x) =  -\frac{1}{12}\log w-\frac{1}{24} \log  (1-108 w)+ \zeta'(-1).
\eeq
For $g\geq2$, $\mathcal{F}^{\{3\}}_g(x)$ has the expression:
\beq\label{triFform1}
\mathcal{F}^{\{3\}}_g(x)=\frac{1-2g}{(2g)!}\,B_{2g}\,\p_x^{2g-2}(\log w) +\sum_{\ell=3g-3}^{5g-5}\frac{a_{g,\ell}}{(1-108w)^{\ell}}\,,
\eeq
where $a_{g,3g-3},\dots,a_{g,5g-5}$ are rational numbers. 
\end{theorem}	

For $g\ge2$, the coefficient $a_{g,5g-5}$ has the expression
\beq\label{tritopcoef}
a_{g,5g-5}=\frac{162^g}{3888}\frac{C_g}{(5g-3)(5g-5)},
\eeq
which can be straightforwardly deduced using~\eqref{triFform1} and the well-known result (see e.g.~\cite{BD}) 
\beq\label{tricorrasymp}
n_g(3^{2j})\sim \frac{16}{\sqrt{3}} \, \frac{\bigl(108\sqrt{3}\bigr)^j}{\bigl(256\sqrt{3}\bigr)^{g}}  \frac{(2j)!\,(2j)^{\frac{5g-7}{2}}}{\Gamma({\frac{5g-1}2})}\,C_g \quad (j\to \infty) \,.
\eeq
Here $C_g$, with $C_0=-1, \ldots$, are constants determined (cf.~e.g.~\cite{JK, Ka, BD,DYZ, Eynard, IZ, YZZ}) by requiring that
the formal series $U=\sum_{g\geq 0}C_g X^{\frac{1-5g}{2}}$ satisfies the Painlev\'e I equation
\beq\label{painleveI}
\frac{d ^2 U}{d X^2}+\frac{1}{16}U^2-\frac{1}{16}X=0\,.
\eeq

Table~\ref{tritable} consists of $n_g(3^{4g-4+2d})$ for $g=0,\dots,4$ and $d=1,\dots,5$.
For the reader's convenience, we also provide
\begin{align}
		\mathcal{F}^{\{3\}}_2&=\frac{\p_x^2 (\log w)}{240}-\frac{351}{8 (1-108 w)^3}+\frac{27}{8 (1-108w)^4}+\frac{189}{10 (1-108 w)^5}\,,
		\label{triF2form1}\\
		\mathcal{F}^{\{3\}}_3&=-\frac{\p_x^4 (\log w)}{6048}+\frac{589761}{4 (1-108 w)^6}-\frac{8203437}{28 (1-108 w)^7}
		-\frac{448335}{2 (1-108 w)^8}+\frac{324405}{(1-108 w)^9}\nn\\ & \quad\quad +\frac{178605}{(1-108 w)^{10}}\,.\label{triF3form1}
\end{align}

\begin{table}[phbt]
\centering
\begin{tabular}{|l|m{2.2cm}|m{2.5cm}|m{2.7cm}|m{3cm}|m{3.7cm}|}
\hline
${\tiny d}$ & $g=0$ & $g=1$ & $g=2$ & $g=3$ & $g=4$  \\ 
\hline
$1$ & $0$ & $3$& $3061800$ & 357485480352000 & 561734730904309522 560000  \\
\hline
$2$ & $0$ & $4536$ & $89414357760$ & 475379823378087 93600 & 208281465835272806 019563520000 \\
\hline
$3$ & $12$ & $19362240$ & 2834113460935 680 & 514591710352541 8098278400 & 547188956214674466 69373094461440000\\
\hline
$4$ & $5184$ & $164367221760$ & 1107578328829 37856000 & 565109847632479 817270034432000 & 131217479838294406
811733692434863882 24000 \\
\hline
$5$  & $9797760$  & 233201956829 1840 &  5405486118155 731877068800 & 672926687318093 573568845764362 24000
& 313985184119369209 780057086172883719 8151680000 \\
\hline
\end{tabular}
\caption{The numbers $n_{g}(3^{4g-4+2d})$ for $g=0,\dots,4$ and $d=1,\dots,5$.} 
\label{tritable}
\end{table}

\medskip

{\bf The case $b=4$.}
According to~\cite{BIZ,DuY2,ELT}, 
$w=x+12x^2+288x^3+\cdots$ is the unique solution to
\beq\label{quadwgenus0eqn}
x =w  -  12 \, w^2\,.
\eeq
Explicitly, $w=\frac{1-\sqrt{1-48x}}{24}$.
\begin{theorem}\label{quadthm1}
For $g=0$, we have $\frac{\p^2 \F^{\{4\}}_0(x)}{\p x^2}=\log w$. For $g=1$,
\beq
\mathcal{F}^{\{4\}}_1(x) =  -\frac{1}{12}\log w-\frac{1}{12} \log (1-24w)+ \zeta'(-1).
\eeq
 For $g\geq 2$, $\F^{\{4\}}_g(x)$ has the expression: 
\beq\label{quadFform}
\F^{\{4\}}_g(x)=\frac{1-2g}{(2g)!}\,B_{2g}\,\p_x^{2g-2}(\log w) +\sum_{\ell=4g-4}^{5g-5}\frac{\tilde{a}_{g,\ell}}{(1-24w)^{\ell}}\,,
\eeq
where $\tilde{a}_{g,4g-4},\dots,\tilde{a}_{g,5g-5}$ are rational numbers. 
\end{theorem}
\noindent The proof, which is based on~\cite{AvM,DuY,GMMMO} and uses a method from~\cite{DYZ}, is given in Section~\ref{quadrangulation}.

It is easy to deduce using~\eqref{quadFform} and a result of~\cite{E2} that the number $\tilde{a}_{g,5g-5}$  equals to $\frac{48^g}{576}\frac{C_g}{(5g-3)(5g-5)}$ (this statement can also be proved using just
the Toda lattice theory), where $C_g$ are the universal constants introduced in~\eqref{tricorrasymp}.

\medskip

{\bf The case $b=2\nu$.} 
In this case, the following theorem (using our notations) was originally conjectured in~\cite[Conjecture~5.2]{ELT}
and confirmed recently in the third arXiv version of~\cite{ELT}.

\smallskip 

\noindent {\bf Theorem A} (\cite{ELT})
{\it For $\nu\geq 2$, define $q=q(x)$ as the power-series-in-$x$ solution to the equation
		\beq\label{qeqn}
		1-q+\frac{(2\nu)!}{\nu!(\nu-1)!} \; x^{\nu-1} q^{\nu}=0\,.
		\eeq 
		Then for $g\geq 2$,
		\beq\label{evenFform4}
		x^{2g-2}\F_g^{\{2\nu\}}(x)
		=\sum_{\ell=2g-2}^{5g-5}\frac{r_{g,\ell}(\nu)}{(\nu-(\nu-1)q)^{\ell}}\,,
		\eeq 
		where for each $\nu\geq 2$, $r_{g,2g-2}(\nu),\dots,r_{g,5g-5}(\nu)$ are rational numbers.}

\smallskip 

We mention here that for the case when $b=4$ one can use $g$ parameters
to express the genus $g$ free energy with $g\ge2$ (see~\eqref{quadFform}) instead of $3g-2$ parameters (see~\eqref{evenFform4}).
For the case when $b=6$ we have similar observations in which $2g-1$ parameters are sufficient, whose proof will be given elsewhere.

We will also give several proofs of Theorem~A, which were completed before noticing the 
updated version of~\cite{ELT}.

In Section~\ref{quadrangulation}, 
a proof of Theorem~A will be achieved based on Theorem~\ref{quadthm1}. 
By using a formula of~\cite{E2}, we will give a second proof, which 
also corresponds to the proof in~\cite{ELT}.

In~\cite{DLYZ20, DuY2} a relationship between $n_g(2\nu_1,\dots,2\nu_k)$ and certain cubic Hodge integrals on the Deligne--Mumford 
moduli spaces, called the {\it Hodge--GUE correspondence}, was established.
As an application of the Hodge--GUE correspondence, we will give a third proof of Theorem~A.

Motivated by a study of Gharakhloo--Latimer~\cite{GL}, we give in the following theorem a fine structure 
for $\F_g^{\{2\nu\}}(x)$ as another application of the Hodge--GUE correspondence.

\begin{theorem}\label{rglthm}
For $g\geq 2$, the coefficients $r_{g,\ell}(\nu)$ in~\eqref{evenFform4} are polynomials in $\nu$ of degree $3g-3$.
\end{theorem}
\noindent The proof is in Section~\ref{evenangulation}.

Following~\cite{GL}, define $S_{g,k}(\nu)$ by 
\beq\label{evenexplicit}
n_{g}((2\nu)^k)
=:\biggl(\frac{(2\nu)!}{(\nu+1)!\nu!}\biggr)^k S_{g,k}(\nu)\,.
\eeq
By~\eqref{evenFform4}, \eqref{qeqn} we have
\begin{align}\label{Sgjdef}
S_{g,k}(\nu)=k!(\nu(\nu+1))^k \sum_{m=0}^{k} \binom{\nu k-1}{k-m}\frac{(\nu-1)^m}{m!} 
\sum_{\ell=2g-2}^{5g-5}r_{g,\ell}(\nu)(\ell-1)_m\,.
\end{align}
Theorem~\ref{rglthm} then implies a conjectural statement given by Gharakhloo--Latimer in~\cite[Remark~2.9]{GL}: 
{\it for $g\geq 0$ and $k\geq 1$, $S_{g,k}(\nu)$ is a polynomial in $\nu$ of degree $3g-3+3k$}.
A further conjectural statement on distributions of zeros of~$S_{g,k}(\nu)$ was also proposed in~\cite{GL}. We recall that,
when $g=0$, the expressions of $n_0((2\nu)^k)$, $\nu\geq 2$, were given in~\cite{EMP} (the $\nu=2$ case also given in~\cite{CL17});
when $g=1$, the expressions of $n_1((2\nu)^k)$ were obtained in~\cite{CL17} for $\nu=2$ and in~\cite{ELT2} for $\nu\ge2$;
when $g=2$, the expressions of $n_2((2\nu)^k)$, $\nu\geq 2$, $k=1,2,3$, were given in~\cite{GL}; 
the expressions of $n_{g}(4^k)$, $g=2,\dots,7$, $k\geq 1$, were given in~\cite{ELT2}, 
and the expressions of $n_{g}(6^k)$, $g=2,\dots,5$, $k\geq 1$, were given in~\cite{GL}. We also note that, 
the above-proved polynomiality of $S_{g,k}(\nu)$ can alternatively be deduced from the celebrated quasi-polynomiality of~\cite{DM, KLS} (cf.~also~\cite{GMM,Nor}), which 
is deeply related to topological recursion.

\smallskip

\noindent {\bf Organization of the paper} 
In Section~\ref{review} we review earlier works on the free energy $\F(x,\bs;\e)$. In Section~\ref{triangulations} we prove Theorem~\ref{trithm1}. In Section~\ref{quadrangulation} we prove Theorem~\ref{quadthm1}. 
In Section~\ref{evenangulation} we prove Theorem~\ref{rglthm}.

\smallskip

\noindent {\bf Acknowledgements} 
We thank Paul Norbury for pointing out the reference~\cite{EO09}.
The work of E.G-F.~is supported by the Ram\'on y Cajal Fellowship RYC2023-045188-I, funded by MCIN/AEI/10.13039\allowbreak/501100011033 and by the FSE+. She also acknowledges support by the project PID2024-155686NB-I00 of the Spanish Ministry of Science and Innovation, the ANR CarteEtPlus ANR-23-CE48-0018, a Tremplin grant from Sorbonne Universit\'e, a PEPS grant from the CNRS and the ERC Synergy Grant ReNewQuantum.
The work of D.Y. and J.X. is supported by NSFC~12371254 and CAS~YSBR-032.

\section{Review of the GUE free energy $\F(x,\bs;\e)$}\label{review}
In this section we review several properties of the GUE free energy $\F(x,\bs;\e)$.

It is known (cf.~e.g.~\cite{GMMMO,HZ,MMMM,Mehta,Mo}) that the partition function $Z=Z(x,{\bf s};\e)$ satisfies the following {\it Virasoro constraints}:
\beq\label{virasoro}
L_{k}(Z(x,\bs;\e))=0\,,\quad k\geq -1\,,
\eeq
where $L_k$ are linear operators given by
\begin{align}
&L_{-1}=\sum_{j\geq 2}j s_j \frac{\p }{\p s_{j-1}}-\frac{\p }{\p s_1}+\frac{x s_1}{\e^2}\,,\label{string}\\
&L_0=\sum_{j\geq 1}j s_j \frac{\p }{\p s_j}-\frac{\p }{\p s_2}+\frac{x^2}{\e^2}\,,\label{scaling}\\
&L_k=\e^2\sum_{j=1}^{k-1}\frac{\p^2 }{\p s_j \p s_{k-j}}+2x\frac{\p }{\p s_k}+\sum_{j\geq 1}j s_j\frac{\p }{\p s_{j+k}}-\frac{\p }{\p s_{k+2}}\,,\quad k\geq 1\,,
\end{align}
which satisfy the Virasoro commutation relations:
\beq
[L_k,L_\ell]=(k-\ell)L_{k+\ell}\,,\quad k,\ell\geq -1\,.
\eeq
The $k=-1$ equation in~\eqref{virasoro} is also known as the {\it string equation}. 
It also follows from~\eqref{matrixintegral} that $Z$ satisfies the following {\it dilaton equation}:
\beq
\sum_{j\geq1} s_j \, \frac{\p Z}{\p s_j}  + \e \, \frac{\p Z}{\p \e} + x \, \frac{\p Z}{\p x} + \frac1{12} \, Z = \frac12 \, \frac{\p Z}{\p s_2}\,.
\label{dilaton}
\eeq

Denote by $\Lambda=e^{\e \p_x}$ the shift operator. 
It is known (cf.~\cite{AvM,DuY,GMMMO,Mehta}) that 
$(V^{\rm GUE}=V^{\rm GUE}(x,\bs;\e),W^{\rm GUE}=W^{\rm GUE}(x,\bs;\e))$
defined by
\begin{align}\label{defVWintro}
V^{\rm GUE}(x,{\bf s};\e) &= \e (\Lambda-1) \frac{\p \F(x,{\bf s};\e)}{\p s_1}\,,\\
W^{\rm GUE}(x,{\bf s};\e) &=\e^2\frac{\p^2 \F(x,{\bf s};\e)}{\p s_1 \p s_1}
\end{align}
is a solution to the {\it Toda lattice hierarchy}~\cite{Flaschka, Manakov}:
\beq
\e\frac{\p L}{\p s_j}=[A_j, L]\,,\quad j\geq 1\,,
\eeq
where
\beq
L=\Lambda+V+W \Lambda^{-1}\,,\quad A_j\coloneqq(L^j)_{+}\,.
\eeq
Here and below, for a difference operator $P=\sum_{k\in\ZZ} P_k \Lambda^k$, $P_{+}\coloneqq\sum_{k\geq 0} P_k \Lambda^k$ and
${\rm res} \, P\coloneqq P_0$. 
Denote 
\beq
V^{\{b\}}(x,\e)\coloneqq V^{\rm GUE}(x,{\bf 1}_b;\e)\,, \quad W^{\{b\}}(x,\e)\coloneqq W^{\rm GUE}(x,{\bf 1}_b;\e)\,.
\eeq
By definition, we know that $W^{\{b\}}(x,\e)$ has the following genus expansion:
\beq\label{topoWbexpand}
W^{\{b\}}(x,\e) = \sum_{g\geq0} \e^{2g} W^{\{b\}}_{g}(x)\,.  
\eeq

Dividing the $k=-1$ equation in~\eqref{virasoro} by $Z(x,\bs;\e)$, we have 
\beq\label{string2}
\sum_{j\geq 2}j s_j \;\frac{\p \F(x,\bs;\e)}{\p s_{j-1}}+\frac{x s_1}{\e^2}
=\frac{\p \F(x,\bs;\e)}{\p s_1}\,.
\eeq
Applying $\e(\Lambda-1)$ on both sides of~\eqref{string2} yields
\beq\label{string21}
\e\sum_{j\geq2}j s_j (\Lambda-1)\frac{\p \F(x,\bs;\e)}{\p s_{j-1}}+s_1=\e(\Lambda-1)\frac{\p \F(x,\bs;\e)}{\p s_1}\,.
\eeq
Following~\cite[(2.1.4)]{DuY}, define
$
h_j\coloneqq\frac1{j+2}{\rm res}L^{j+2}\in\QQ[V,W,\Lambda^{\pm 1}V,\Lambda^{\pm 1}W,\dots]
$. 
In particular, $h_{-1}=V$. 
From~\cite[(1.2.9) and Corollary~2.2.2]{DuY} we know 
\beq
\e(\Lambda-1)\frac{\p \F(x,\bs;\e)}{\p s_j}=j h_{j-2}\bigl|\bigr._{\Lambda^i V=\Lambda^i V^{\rm GUE},\,\Lambda^i W=\Lambda^i W^{\rm GUE},\, i\in\ZZ} \, ,
\eeq
then~\eqref{string21} becomes
\beq\label{Ptype1}
\sum_{j\geq 2}j(j-1) s_j \; h_{j-3}\bigl|\bigr._{\Lambda^i V=\Lambda^i V^{\rm GUE},\,\Lambda^i W=\Lambda^i W^{\rm GUE},\, i\in\ZZ}
+s_1
=V^{\rm GUE}\,.
\eeq
Applying $\e^2\frac{\p}{\p{s_1}}$ on both sides of~\eqref{string2}, we get 
\beq\label{string22}
\e^2\sum_{j\geq2}j s_j \frac{\p^2 \F(x,\bs;\e)}{\p s_1\p s_{j-1}}+x=\e^2\frac{\p^2 \F(x,\bs;\e)}{\p s_1^2}\,.
\eeq
Then according to~\cite[(1.2.8), (2.1.10)]{DuY} and~\cite[(56)]{Y}, we obtain
\beq\label{Ptype2}
\sum_{j\geq2}j s_{j}(\Lambda+1)^{-1}(j h_{j-2}-(j-1)V h_{j-3}) |_{\Lambda^i V=\Lambda^i V^{\rm GUE},\,\Lambda^i W=\Lambda^i W^{\rm GUE},\, i\in\ZZ}+x
=W^{\rm GUE}\,.
\eeq
We note that special cases of~\eqref{Ptype1} and~\eqref{Ptype2} were given in e.g.~\cite{BD,BK,CL17,fik91}.

Define
${\bf v}^{\rm GUE}(x,{\bf s})=(v^{\rm GUE}(x,{\bf s}),u^{\rm GUE}(x,{\bf s}))\coloneqq(V^{\rm GUE}(x,{\bf s};0), \log W^{\rm GUE}(x,{\bf s};0))$.
Based on the Dubrovin--Zhang theory, the following formulas were obtained in~\cite{Du2} (cf.~also~\cite{ACKM, DW,Eynard,Y2}):

\begin{align}
	\mathcal{F}_0(x,{\bf s}) = &
	\frac12\sum_{p,q\geq0} (p+1)! (q+1)! \bigl(s_{p+1}-\frac12\delta_{p,1}\bigr) \bigl(s_{q+1}-\frac12\delta_{q,1}\bigr) 
	\Omega^{{\mathbb{P}^1},[0]}_{2,p;2,q}(v^{\rm GUE}(x,{\bf s}),u^{\rm GUE}(x,{\bf s}))   \nn\\
	& + x\sum_{p\geq0} (p+1)! \bigl(s_{p+1}-\frac12\delta_{p,1}\bigr) \theta^{\mathbb{P}^1}_{2,p}(v^{\rm GUE}(x,{\bf s}),u^{\rm GUE}(x,{\bf s})) + \frac12 x^2 u^{\rm GUE}(x,{\bf s})\,, \label{F0dub}\\
	\F_g(x,{\bf s}) = 
	&F_g^{\mathbb{P}^1}\biggl(v^{\rm GUE}(x,{\bf s})\,,
	u^{\rm GUE}(x,{\bf s}), \dots,\frac{\p^{3g-2} v^{\rm GUE}(x,{\bf s})}{\p x^{3g-2}},\frac{\p^{3g-2} u^{\rm GUE}(x,{\bf s})}{\p x^{3g-2}}\biggr) \nn\\
	& + \bigl(\zeta'(-1)-\frac1{24} \log(-1)\bigr)\delta_{g,1}\,,\quad g\geq1\,.\label{fgfmgequal}
\end{align}
Here, $\Omega^{{\mathbb{P}^1},[0]}_{2,p;2,q}(v,u)$ and $\theta^{{\mathbb{P}^1}}_{2,p}(v,u)$ 
are certain genus $0$ two-point correlation functions of the
$\mathbb{P}^1$-Frobenius manifold, and $F_g^{\mathbb{P}^1}(v,u,v_1,u_1,\dots,v_{3g-2},u_{3g-2})$ denotes the genus~$g$ free energy in jet variables.
For $g=1$,
\beq\label{jetF1}
F_1^{\mathbb{P}^1}(v,u,v_1,u_1)=\frac1{24}\log(v_1^2-e^{u}u_1^2)-\frac1{24}u\,.
\eeq
For $g\geq 2$, $F_g^{\mathbb{P}^1}(v,u,v_1,u_1,\dots,v_{3g-2},u_{3g-2})$ depends polynomially on $v_2,u_2,\dots,v_{3g-2},u_{3g-2}$ and rationally on $v_1,u_1$ with coefficients being smooth functions of $v$ and $u$. Moreover, 
\beq
\deg F^{\mathbb{P}^1}_{g}(v,u,v_1,u_1,\dots,v_{3g-2},u_{3g-2})=2g-2\,,
\eeq
where $\deg v_k=\deg u_k\coloneqq k$.
The reader is referred to e.g.~\cite{Du96,DuZ,Y2} for details on the $\mathbb{P}^1$-Frobenius manifold, $\theta^{{\mathbb{P}^1}}_{2,p}(v,u)$,
$\Omega^{{\mathbb{P}^1},[0]}_{2,p;2,q}(v,u)$, and $F_g^{\mathbb{P}^1}(v,u,v_1,u_1,\dots,v_{3g-2},u_{3g-2})$. 

Following for example~\cite{W}, consider {\it the even GUE free energy} $\F^{\rm even}$, which is the GUE free energy restricted to even couplings, i.e.,
\beq
\F^{\rm even} = \F^{\rm even}(x,s_2,s_4,\dots;\e)\coloneqq\F(x,\bs;\e)|_{s_1=s_3=\cdots=0}
\eeq
and plays an important role in two-dimensional quantum gravity~\cite{W}.
Denote by 
$\F^{{\rm even}}_g\coloneqq[\e^{2g-2}]\F^{\rm even}$ the genus~$g$ part of $\F^{\rm even}$, and denote
\beq
 v_{\rm even}=v^{\rm GUE}(x,\bs)|_{s_1=s_3=\dots=0}\,,\;
 u_{\rm even}=u^{\rm GUE}(x,\bs)|_{s_1=s_3=\dots=0}\,. 
\eeq
It was proved in~\cite{DuY2} that $v_{\rm even}\equiv 0$ and
$w_{\rm even}\coloneqq e^{u_{\rm even}}$ satisifies
\beq\label{evengenus0eqn}
x=w_{\rm even}-\sum_{\nu\geq 1}\frac{(2\nu)!}{\nu!(\nu-1)!}s_{2\nu} w_{\rm even}^{\nu}\,.
\eeq

The {\it Hodge--GUE correspondence}, which gives an explicit relationship between 
a certain cubic Hodge free energy and $\F^{\rm even}$,
was established in~\cite{DLYZ20,DuY2}.
Let
$\overline{\mathcal{M}}_{g,n}$ denote the Deligne--Mumford moduli space of stable genus-$g$ curves with $n$ marked points. Denote by $\psi_i$ the first Chern class of the $i$th tautological line bundle 
over $\overline{\mathcal{M}}_{g,n}$, and by $\lambda_i$ the $i$th Chern class of the rank-$g$ Hodge bundle
on $\overline{\mathcal{M}}_{g,n}$.

We define the Chern polynomial $\Lambda_g(z)\coloneqq\sum_{i=0}^g \lambda_i z^i$, and set
\beq
\mathcal{H}({\bf t};\e)\coloneqq\sum_{g\geq 0}\e^{2g-2}\sum_{n\geq 0}\frac1{n!}\sum_{i_1,\cdots,i_n\geq 0} \prod_{m=1}^{n}t_{i_m}\int_{\overline{\mathcal{M}}_{g,n}}\Lambda_g(-1)\Lambda_g(-1)\Lambda_g(\tfrac12)
\psi_1^{i_1}\cdots\psi_{n}^{i_n}\,.
\eeq
The Hodge--GUE correspondence says that 
\beq
\F^{\rm even}(x,s_2,s_4,\dots)=
\e^{-2}A(x,s_2,s_4,\dots)+\zeta'(-1)
+(\Lambda^{\frac12}+\Lambda^{-\frac12})\mathcal{H}({\bf t}(x,s_2,s_4,\dots);\e)\,,
\eeq 
where $A=A(x,s_2,s_4,\dots)$ is defined by
\beq
A=\frac12\sum_{k_1,k_2\geq 1}\frac{k_1 k_2}{k_1+k_2}
\binom{2k_1}{k_1}\binom{2k_2}{k_2}s_{2k_1}s_{2k_2}
-\sum_{k\geq 1}\frac{k}{1+k}\binom{2k}{k}s_{2k}+x\sum_{k\geq1}\binom{2k}{k}s_{2k}+\frac14-x\,,
\eeq
and
\beq
t_i(x,s_2,s_4,\dots)\coloneqq\sum_{k\geq 1} k^{i+1}\binom{2k}{k}s_{2k}-1+\delta_{i,1}+x \delta_{i,0}\,,\quad i\geq 0\,.
\eeq
This correspondence yields the only known ELSV-type formula for strictly monotone Hurwitz numbers~\cite{BG17}, to the best of our knowledge.
Explicit expressions for $\F^{{\rm even}}_g$ for $g=1,\dots,5$ 
in terms of $u_{\rm even}$ and its $x$-derivatives were given in~\cite{DuY2} and its arXiv 
preprint version. For example, \cite{DuY2} shows that
\beq\label{Feven1general}
\F^{{\rm even}}_1= \frac1{12} \log  \frac{\p u_{\rm even}}{\p x} +\zeta'(-1).
\eeq

\section{On triangulations}\label{triangulations}
In this section we prove Theorem~\ref{trithm1}.

Taking $\bs={\bf 1}_3$ in~\eqref{Ptype1} and \eqref{Ptype2}, we have
\begin{align}
3 \, \bigl(V^{\{3\}}(x,\e)^2 + W^{\{3\}}(x,\e) + W^{\{3\}}(x+\e,\e)\bigr) & =  V^{\{3\}}(x,\e) \,, \label{eq1}\\ 
x + 3 \, W^{\{3\}}(x,\e) \, \bigl(V^{\{3\}}(x,\e)+V^{\{3\}}(x-\e,\e)\bigr) & = W^{\{3\}}(x,\e)  \label{eq2}
\end{align}
(cf.~also~e.g.~\cite{BD,CL17}). Introduce 
\beq
\widetilde V^{\{3\}}(x,\e) \coloneqq V^{\{3\}}\bigl(x-\frac{\e}2,\e\bigr)\,,
\eeq
then equations~\eqref{eq1}--\eqref{eq2} become
\begin{align}
 3 \,  \Bigl(\widetilde V^{\{3\}}(x,\e)^2 + W^{\{3\}}\bigl(x-\frac\e2,\e\bigr) \, +\, W^{\{3\}}\bigl(x+\frac\e2,\e\bigr)\Bigr) & =  \widetilde V^{\{3\}}(x,\e) \,, \label{eq1t}\\ 
 x + 3  W^{\{3\}}(x,\e)  \Bigl(\widetilde V^{\{3\}}\bigl(x-\frac\e2,\e\bigr)+\widetilde V^{\{3\}}\bigl(x+\frac\e2,\e\bigr)\Bigr) & = W^{\{3\}}(x,\e) \,. \label{eq2t}
\end{align}

From the definition of $V^{\rm GUE}(x,\bs;\e)$, we know that $\widetilde V^{\{3\}}(x,\e)$ has the following genus expansion:
\begin{align}
\widetilde V^{\{3\}}(x,\e) = \sum_{g\geq0} \e^{2g}  \widetilde V^{\{3\}}_g(x). \label{topoVt} 
\end{align}
Substituting~\eqref{topoWbexpand} and~\eqref{topoVt} in equations~\eqref{eq1t}--\eqref{eq2t}, we obtain
\begin{align}
	& 3 \, \sum_{g_1,g_2\geq0} \e^{2(g_1+g_2)} \, \widetilde V^{\{3\}}_{g_1} \, \widetilde V^{\{3\}}_{g_2} + 6\, \sum_{g,k}  \frac{\e^{2g+2k}}{2^{2k} (2k)!} 
	\frac{\p^{2k} W^{\{3\}}_{g}}{\p x^{2k}} =  \sum_{g\geq0} \e^{2g} \, \widetilde V^{\{3\}}_{g}(x) \,, \label{eq1te}\\ 
	& x + 6 \, \sum_{g_1,g_2,k\geq0}   \frac{\e^{2g_1+2g_2+2k}}{2^{2k} (2k)!} W^{\{3\}}_{g_1}(x)  \,
	\frac{\p^{2k} \widetilde V^{\{3\}}_{g_2}}{\p x^{2k}} =  \sum_{g\geq0} \e^{2g} \, W^{\{3\}}_{g}(x) \,. \label{eq2te}
\end{align}
Comparing coefficients of $\e^{2g}$ ($g \geq 0$) on both sides of~\eqref{eq1te}--\eqref{eq2te}, we find
\begin{align}
	3 \, \widetilde V^{\{3\}}_{0}(x)^2 + 6 \, W^{\{3\}}_{0}(x) & = \widetilde V^{\{3\}}_{0}(x) \,, \label{genus01}\\ 
	x + 6 \, W^{\{3\}}_{0}(x)  \, \widetilde V^{\{3\}}_{0}(x) & =  W^{\{3\}}_{0}(x) \,, \label{genus02}
\end{align}
and for $g\geq1$,
\begin{align}
& 
\bigl(1 - 6 \widetilde V^{\{3\}}_{0} \bigr) \, \widetilde V^{\{3\}}_{g} \,-\, 6 \, W^{\{3\}}_{g} 
=
3 \, \sum_{g_1=1}^{g-1} \widetilde V^{\{3\}}_{g_1} \, \widetilde V^{\{3\}}_{g-g_1} + 6 \, \sum_{k=1}^g  \frac{1}{2^{2k}(2k)!} 
\frac{\p^{2k} W^{\{3\}}_{g-k}}{\p x^{2k}}  \,, \label{genusg1}\\ 
& 
- \, 6 \, W^{\{3\}}_{0} \, \widetilde V^{\{3\}}_{g}  + (1- 6 \, \widetilde V^{\{3\}}_{0}) \, W^{\{3\}}_{g} 
=
6 \, \sum_{g_1,g_2\leq g-1 \atop g_1+g_2+k=g}   \frac{1}{2^{2k} (2k)!} W^{\{3\}}_{g_1} \,
\frac{\p^{2k} \widetilde V^{\{3\}}_{g_2}}{\p x^{2k}}  \,. \label{genusg2}
\end{align}

Solving~\eqref{genus01}--\eqref{genus02} we find that $v=v(x)\coloneqq\widetilde V^{\{3\}}_{0}(x)$ and $w=w(x)\coloneqq W^{\{3\}}_{0}(x)$
satisfy~\eqref{triveqn} and~\eqref{triweqn}, which imply~\eqref{triweqn2} as well as
\beq\label{vwithw}
v=\frac{1-\sqrt{1-72w}}{6}\,.
\eeq
From~\eqref{triweqn2} we get
\beq\label{tridw}
\p_x = \frac{\sqrt{1-72\,w}}{1-108\,w} \, \p_w\,.
\eeq
\begin{lemma}\label{lemDx}
For $k\geq 1$,
\beq\label{dxeven}
\p_x^{2k}=\sum_{i=1}^{2k}\frac{T_{k,i}(w)}{(1-108w)^{4k-i}}\p_w^i 
\eeq
for some polynomials $T_{k,1}(w),\dots,T_{k,2k}(w)$. 
Moreover, $\deg\,T_{k,i}(w)\leq k$ for $i=1,\dots,2k$.
\end{lemma}

\begin{proof}
For $k=1$, from~\eqref{tridw} we obtain
\beq\label{tridw2}
\p_x^2
=\frac{72 (1-54 w)}{(1-108 w)^3}\p_w+\frac{1-72 w}{(1-108 w)^2}\p_w^2\,,
\eeq  
thus~\eqref{dxeven} holds. Suppose that~\eqref{dxeven} is true for $k=m$ with $\deg\,T_{m,i}(w)\leq m$. Then for $k=m+1$ we have
\begin{align}
\p_x^{2m+2}&=\p_x^2\circ\p_x^{2m}=
\biggl(\frac{72 (1-54 w)}{(1-108 w)^3}\p_w+\frac{(1-72 w)}{(1-108 w)^2}\p_w^2 \biggr)
\circ \sum_{i=1}^{2m}\frac{T_{m,i}(w)}{(1-108w)^{4m-i}}\p_w^i\nn\\
&=\sum_{i=1}^{2m+2}\frac{T_{m+1,i}(w)}{(1-108w)^{4m+4-i}}\p_w^i\,.
\end{align}
Here
\begin{align}\label{Tm+1}
T_{m+1,i}=&(1-72w)(1-108w)^2 \p_w^2 T_{m,i}+2(1-72w)(1-108w)\p_w T_{m,i-1}+(1-72w)T_{m,i-2}\nn\\
&+\bigl(72+216(4k-i)-(3888-15552(4k-i))w\bigr)(1-108w)\p_w T_{m,i}\nn\\
&+\bigl(72+216(4k+1-i)-(3888-15552(4k+1-i))w\bigr)T_{m,i-1}\nn\\
&+108(4k-i)\bigl(72+108(4k+1-i)-(3888-7776(4k+1-i))w\bigr)T_{m,i}\,,
\end{align}
with $T_{m,-1}$, $T_{m,0}$, $T_{m,2m+1}$ and $T_{m,2m+2}$ defined as~0. 
Obviously,
$\deg\,T_{m+1,i}\leq m+1$ for $i=1,\dots,2m+2$. By mathematical induction the lemma is proved.
\end{proof}

\begin{prop}\label{trisolansatz}
For $g\geq 1$, $\widetilde{V}^{\{3\}}_{g}(x)$, $W^{\{3\}}_{g}(x)$ are given by
\beq\label{ansatzOftriSol}
W^{\{3\}}_{g}(x)=\frac{w Q_g(w)}{(1-108 w)^{5g-1}}\Big|_{w=w(x)}\, ,\quad \widetilde{V}^{\{3\}}_{g}(x)=\frac{\sqrt{1-72w}\, R_g(w)}{(1-108 w)^{5g-1}}\Big|_{w=w(x)} \,,
\eeq
where $Q_g(w), R_g(w)$ are polynomials of~$w$. Moreover, $\deg\,Q_g(w)\leq 2g-1$, $\deg\,R_g(w)\leq 2g-1$.
\end{prop}
\begin{proof}
For $g=1$, solving the $g=1$ equations of \eqref{genusg1}--\eqref{genusg2} gives
\beq
\widetilde V^{\{3\}}_{1}(x) = \frac{54 \sqrt{1-72 w}}{(1-108 w)^4}\Bigl|_{w=w(x)}\Bigr.\,,\quad
W^{\{3\}}_{1}(x) = \frac{162 \,(5-324 w)\, w}{(1-108 w)^4}\Bigl|_{w=w(x)}\Bigr. \,.
\eeq
Suppose~\eqref{ansatzOftriSol} is true with $\deg\,Q_g(w),\deg\,R_g(w)\leq 2g-1$ for $g\leq m$. 
Then for $g=m+1$, 
by solving the $g=m+1$ equations of~\eqref{genusg1}--\eqref{genusg2} and using~\eqref{vwithw}, Lemma~\ref{lemDx},
we obtain
\begin{align}
(\widetilde{V}^{\{3\}}_{m+1}(x),W^{\{3\}}_{m+1}(x))
=\Bigl(\frac{\sqrt{1-72w}\,R_{m+1}(w)}{(1-108w)^{5m+4}},
\frac{w\,Q_{m+1}(w)}{(1-108w)^{5m+4}}\Bigr)\Bigl|_{w=w(x)}\Bigr.\,,
\end{align}
where $Q_{m+1}(w)$, $R_{m+1}(w)$ are polynomials of $w$ with
$\deg\,Q_{m+1}(w),\deg\,R_{m+1}(w)\leq 2m+1$. This completes the proof of the proposition by mathematical induction.
\end{proof}

From Proposition~\ref{trisolansatz} we know that $\widetilde{V}^{\{3\}}(x,\e)$, $W^{\{3\}}(x,\e)$ can be written as
\beq\label{trigenusinpart1}
W^{\{3\}}(x,\e)=w+\sum_{g\geq 1}\e^{2g}\sum_{\ell=3g-1}^{5g-1}\frac{A_{g,\ell}}{(1-108w)^{\ell}}\,,
\quad \widetilde{V}^{\{3\}}(x,\e)=v+\sum_{g\geq 1}\e^{2g}\sum_{\ell=-\infty}^{5g-1}\frac{\widetilde{B}_{g,\ell}}{(1-108w)^{\ell}}\,,
\eeq
for some real numbers $A_{g,5g-1},\ldots,A_{g,3g-1},\widetilde{B}_{g,5g-1},\ldots$. 
Substituting~\eqref{trigenusinpart1} in~\eqref{eq1t}--\eqref{eq2t}, one can obtain that
$
U=2^{\frac{9}5}3^{\frac{11}5}\sum_{g\geq 0}A_{g,5g-1}(2^{\frac{2}5}3^{\frac85}X)^{\frac{1-5g}{2}}
$
satisfies the Painlev\'e I equation~\eqref{painleveI}. Therefore,
\beq\label{univ_const}
A_{g,5g-1}=\frac{162^g}{108}C_g\,,
\eeq
with $C_g$ the constants introduced in~\eqref{tricorrasymp}.

We proceed to calculate $\F^{\{3\}}(x,\e)$. Taking $\bs={\bf 1}_3$ in~\eqref{defVWintro} we have
\beq
\log W^{\{3\}}(x,\e) = \F^{\{3\}}(x+\e,\e) \,+\, \F^{\{3\}}(x-\e,\e) \,-\, 2 \, \F^{\{3\}}(x,\e) \,. \label{WF}
\eeq
Therefore, 
\begin{align}\label{wf}
\F^{\{3\}}(x,\e) & = \frac{1}{\e^2 \p_x^2} \, \frac{\e^2 \p_x^2}{e^{\e \p_x} + e^{-\e \p_x} -2} \bigl( \log W^{\{3\}}(x,\e) \bigr)\nn \\
&= \frac{1}{\e^2 \p_x^2} \, \sum_{k\geq 0} \frac{(1-2k) B_{2k}}{(2k)!} \, \e^{2k} \p_x^{2k} \bigl( \log W^{\{3\}}(x,\e) \bigr) \,.
\end{align}

It is helpful to introduce the following lemma for computing~\eqref{wf}:
\begin{lemma}\label{lem1}
The following identity holds for $k\geq 4$:
\begin{align}\label{lem1gen}
\frac{972(-2)^k}{(1-108w)^k}\frac{((k-2)!)^2}{(2k-4)!}=
\p_x^{2}\biggl(72w & \sqrt{3-216w} \; {\rm arctanh}\bigl(\sqrt{3-216w}\bigr)\nn\\
&-72w+\sum_{i=1}^{k-4}\frac{1}{(1-108w)^i}\frac{(-2)^{i+2}}{i(i+2)}\frac{(i+2)!^2}{(2i+4)!}\biggr)\,.
\end{align}
\end{lemma}
\begin{proof}
Let $a$ be a formal variable. We have
\beq
972\sum_{k\geq 4}\frac{((k-2)!)^2}{(2k-4)!}\frac{(-2a)^k}{(1-108w)^k}=
3888\, s^2\biggl(\frac{s(s-1)}{2-s}+\frac{2\sqrt{s}\, {\rm arcsin}\bigl(\sqrt{s/2}\bigr)}{(2-s)^{\frac32}}\biggr)\,,
\eeq
where $s=\frac{a}{108w-1}$. Using~\eqref{tridw2} it is easy to verify that
\begin{align}\label{lemstep1}
&3888s^2\biggl(\frac{s(s-1)}{2-s}+\frac{2\sqrt{s}\, 
{\rm arcsin}\bigl(\sqrt{s/2}\bigr)}{(2-s)^{\frac32}}\biggr)\nn\\
&=\p_x^2\Biggl(\frac{a^4}{1-a}\biggl(72w\sqrt{3-216w}\,{\rm arctanh}\,\bigl(\sqrt{3-216w}\bigr)\nn\\
& \;\;\;\;\;\;\;\;\;\;\;\;\;\;\;\;-72w
+\frac23\sqrt{1-\frac{s}2}\,\frac{s+1}{s}\,\frac{{\rm arcsin}\bigl(\sqrt{s/2}\bigr)}{\sqrt{s/2}}
-\frac2{3s}\biggr)\Biggr)\nn\\
&=\p_x^2\Biggl(\frac{a^4}{1-a}\biggl(72w\sqrt{3-216w}\,{\rm arctanh}\bigl(\sqrt{3-216w}\bigr)
\nn\\
& \;\;\;\;\;\;\;\;\;\;\;\;\;\;\;\;-72w
+\sum_{i\geq 0} \frac{a^i}{(1-108w)^i}\frac{(-2)^{i+2}}{i(i+2)}\frac{(i+2)!^2}{(2i+4)!}
\biggr)\Biggr)\,.
\end{align}
In this last equality we use the identity
\beq
\frac23\sqrt{1-\frac{s}2}\,\frac{s+1}{s}\,\frac{{\rm arcsin}\bigl(\sqrt{s/2}\bigr)}{\sqrt{s/2}}=\frac2{3s}+
\frac59+\sum_{i\geq 1}\frac{(-2)^{i+2}}{i(i+2)}\frac{(i+2)!^2}{(2i+4)!} s^i\,.
\eeq
By taking the coefficient of $a^k$, $k\geq 4$, on both sides of~\eqref{lemstep1} we obtain~\eqref{lem1gen}.
\end{proof}

Using Proposition~\ref{trisolansatz}, we have
\begin{align}\label{triWstep2}
\log \, W^{\{3\}}(x,\e)
&=\log \, w \,+\, \log \, \biggl(1\,+\,\sum_{g\geq1}  \, \frac{\e^{2g} Q_g(w)}{(1-108w)^{5g-1}}\biggr) \,. 
\end{align}
By substituting~\eqref{triWstep2} in~\eqref{wf}, and using Lemma~\ref{lemDx} and Lemma~\ref{lem1},
we obtain for $g\geq 2$,  
\begin{align}
	\F^{\{3\}}_g =& \frac{1-2g}{(2g)!}B_{2g}\p_x^{2g-2}(\log w) \nn\\
	&+\sum_{k=1}^{g-1} \frac{1-2k}{(2k)!}B_{2k}\p_x^{2k-2}
	\biggl(\sum_{\ell=1}^{g-k}\frac{(-1)^{\ell-1}}{\ell (1-108w)^{5g-5k-\ell}}
	\sum_{m_1,\dots,m_{\ell}\geq 1\atop m_1+\cdots+m_{\ell}=g-k}\prod_{i=1}^{\ell}Q_{m_i}(w)\biggr)\nn\\
	&+\p_x^{-2}
	\biggl(\sum_{\ell=1}^{g}\frac{(-1)^{\ell-1}}{\ell (1-108w)^{5g-\ell}}
	\sum_{m_1,\dots,m_{\ell}\geq 1\atop m_1+\cdots+m_{\ell}=g}\prod_{i=1}^{\ell}Q_{m_i}(w)\biggr)\label{FgAnsatz1} \\
	=&\frac{1-2g}{(2g)!}B_{2g}\p_x^{2g-2}(\log w) +\sum_{\ell=3g-3}^{5g-5}\frac{a_{g,\ell}}{(1-108w)^{\ell}}
	+\alpha_g w\sqrt{1-72w}+\beta_g\nn\\
	&+\gamma_g \cdot \biggl(72w \sqrt{3-216w}\, {\rm arctanh}(\sqrt{3-216w})\nn\\
	&-72z+\sum_{i=1}^{3g-4}\frac{1}{(1-108w)^i}\frac{(-2)^{i+2}}{i(i+2)}\frac{(i+2)!^2}{(2i+4)!}\biggr)\,.\label{FgAnsatz2}
\end{align}
Here $\alpha_g$ and $\beta_g$ are integration constants, and $a_{g,3g-3},\dots, a_{g,5g-5},\gamma_g$ are rational numbers. 

For further simplification of~\eqref{FgAnsatz2}, introduce  
\beq\label{Gdef}
\mathcal{G}^{\{3\}}(x,\e) \coloneqq 6 \, \Bigl(\frac14 \,-\, \frac{x^2}{\e^2} + 3 \, (\e \p_\e + x \p_x) \, (\mathcal{F}^{\{3\}}(x,\epsilon))\Bigr)\,.
\eeq
It has the genus expansion $\mathcal{G}^{\{3\}}(x,\e)=\sum_{g\geq 0}\e^{2g-2}\mathcal{G}^{\{3\}}_g(x)$,
and for $g\geq 0$,
\beq\label{GandF}
\mathcal{G}^{\{3\}}_g=-6x^2 \delta_{g,0}+\frac{3}{2}\delta_{g,1}+18(2g-2)\F^{\{3\}}_g+18x\p_{x}(\F^{\{3\}}_g)\,.
\eeq

By taking $\bs={\bf 1}_3$ in equations~\eqref{string}, \eqref{scaling} and~\eqref{dilaton}, we find 
$\mathcal{G}^{\{3\}}(x,\e) =  \frac{\p Z}{\p s_1}$,
so from~\eqref{defVWintro} we know
\beq
\frac1\e \, \widetilde V^{\{3\}}(x,\e) = \mathcal{G}^{\{3\}}(x+\e/2,\e) - \mathcal{G}^{\{3\}}(x-\e/2,\e) \,. \label{VF}
\eeq
Therefore
\beq
\mathcal{G}^{\{3\}}(x,\e) 
= \frac1{\e^2\p_x}  \frac{\e\p_x} {(\Lambda^{1/2}-\Lambda^{-1/2})} (\widetilde V^{\{3\}}(x,\e)) 
= \p_x^{-1}  \sum_{g\geq0} \e^{2g-2} \sum_{k+g_1=g} \frac{(-1)^k B_{2k}(\frac12)}{(2k)!}\p_x^{2k} (\widetilde V^{\{3\}}_{g_1} (x))\,, \label{triGstructure}
\eeq
and
\beq\label{triGg}
\mathcal{G}_g^{\{3\}}(x,\e)=\p_x^{-1}(\widetilde{V}^{\{3\}}_{g}(x))+\sum_{k=1}^g \frac{(-1)^k B_{2k}(\frac12)}{(2k)!}\p_x^{2k-1} (\widetilde{V}^{\{3\}}_{g-k} (x))\,,\quad g\geq 0\,.
\eeq
Here $B_k(p)$ denotes the $k$th Bernoulli polynomial.

\begin{prop}\label{Ggprop}
For $g\geq2$, $\mathcal{G}^{\{3\}}_g(x)$ is given by
\beq\label{triGgstructure}
\mathcal{G}^{\{3\}}_g(x)= \frac{S_g(w)}{(1-108w)^{5g-3}}\Bigl|_{w=w(x)}\Bigr.+\tilde{\lambda}_g\,,
\eeq 
for a polynomial $S_g(w)$ and some integration constant $\tilde{\lambda}_g$.
Moreover, $\deg\, S_g(w)\leq 2g-1$.
\end{prop}
\begin{proof}
From~\eqref{triweqn2} we know that $\p_x^{-1}(*)=\p_w^{-1}\bigl(\frac{1-108 w}{\sqrt{1-72 w}}\,*\bigr)$.
Using this and Proposition~\ref{trisolansatz}, Lemma~\ref{lemDx} in~\eqref{triGg}, we find expression~\eqref{triGgstructure} and that $\deg\, S_g(w)\leq 2g-1$.
\end{proof}

We are ready to prove Theorem~\ref{trithm1}.
\begin{proof}[Proof of Theorem~\ref{trithm1}]
The $g=0$ case has been given above (cf.~also~\eqref{F0dub} and \cite{DuY}).
For $g=1$, the statement can be verified directly using~\eqref{genusg1}--\eqref{genusg2} (cf.~also~\eqref{fgfmgequal}--\eqref{jetF1} and \cite{DuY}).
For $g\ge2$, by comparing~\eqref{triFform1} and~\eqref{FgAnsatz2}, it suffices to prove that the constants
$\alpha_g, \beta_g, \gamma_g$ appearing in~\eqref{FgAnsatz2} all vanish.
From~\eqref{tridw2} we know that $\p_x^{2g-2}(\log w)$ is a rational function of $w$.
Putting~\eqref{FgAnsatz2} in~\eqref{GandF} and comparing with~\eqref{triGgstructure} in
Proposition~\ref{Ggprop}, by
the vanishing of ${\rm arctanh}\,(\sqrt{3-216w})$ in $\mathcal{G}^{\{3\}}_g$ we find $\gamma_g=0$. 
Since $\p_x (\log w)=\frac{\sqrt{1-72 w}}{(1-108 w) w}$, using~\eqref{vwithw}--\eqref{tridw} we find that for $k\geq 1$,
\beq
\p_x^k (v) = O(w^{\frac12-\frac32 k})\,, \quad 
\p_x^k (\log w) = O(w^{-\frac32 k})\,, \quad w \to \infty\,.
\eeq
Then from~\eqref{fgfmgequal} (cf.~\cite{YZ}) we know that 
$\F^{\{3\}}_g \to 0$ as $w\to\infty$. Thus $\alpha_g=\beta_g=0$.
\end{proof}

We end this section by presenting several consequences of Theorem~\ref{trithm1}.
First of all, note that, using Lemmas~\ref{lemDx}, \ref{lem1}, \eqref{trigenusinpart1}, \eqref{FgAnsatz1}, \eqref{FgAnsatz2} and Theorem~\ref{trithm1}, we can achieve a new self-contained proof of the expression~\eqref{tritopcoef} of $a_{g,5g-5}$,
which in turn implies~\eqref{tricorrasymp}.

Multiplying both sides of~\eqref{triFform1} by 
$dx/x^{3-2g+j}$ and taking the residue at $x=0$, we obtain 
\begin{align}
	n_{g}(3^{2j})&=\sigma_{g,j}+72^{j-2g+2}(2j)!\sum_{m=0}^{j-2g+2} 
	\binom{\frac{3j}2-3g+3-m}{\frac{j}2-g+1}
\sum_{\ell=3g-3}^{5g-5}(\ell-1)_{m}\,\frac{3^m}{2^m m!}a_{g,\ell}\,,
\end{align}
where 
$
\sigma_{g,j}=(2j)!\mathop{\rm res}\limits_{w=0}\frac{(1-108w)dw}{w^{j-2g+3}(1-72w)^{\frac{j}2-g+2}}\frac{(1-2g)B_{2g}}{2g!}\p_x^{2g-2}(\log w)
$
and $(a)_m\coloneqq a(a+1)\cdots (a+m-1)$ is the Pochhammer symbol.

Following~\cite{E} (cf.~\cite{Eynard}), introduce
\beq\label{triqchange}
q=\frac1{1-72w}\,.
\eeq
Formulas \eqref{triF1form1333},~\eqref{triF2form1} are then translated to
\begin{align}
	\F^{\{3\}}_1 =-\frac1{12}&\log x-\frac1{24} \log \biggl(\frac{3-q}2\biggr)+\zeta'(-1)\, ,\label{triF12inq}\\
	\F^{\{3\}}_2 =\frac1{x^2}\biggl(&\frac{3}{64 \left(3-q\right)}-\frac{29}{32 \left(3-q\right)^2}+\frac{191}{48 \left(3-q\right)^3}-\frac{55}{8 \left(3-q\right)^4}+\frac{21}{5 \left(3-q\right)^5}\,\biggr)\,,\label{triF2inq}\\
	\F^{\{3\}}_3  =\frac1{x^4}\biggl(&\frac{63}{256 \left(3-q\right)^2}-\frac{22765}{1152 \left(3-q\right)^3}+\frac{7925}{24 \left(3-q\right)^4}-\frac{39311}{16 \left(3-q\right)^5}\nn\\
	&+\frac{1443995}{144 \left(3-q\right)^6}
	-\frac{4055053}{168\left(3-q\right)^7}+\frac{68625}{2\left(3-q\right)^8}-\frac{26730}{\left(3-q\right)^9}+\frac{8820}{\left(3-q\right)^{10}}\biggr) \,. \label{triF3inq}
\end{align}
Formula~\eqref{triF2inq} was given in~\cite{E2} (cf.~\cite{Eynard}), and~\eqref{triF3inq} was given in~\cite{Eynard}.
It was proved by Eynard~\cite{Eynard} (using a slightly different notation) that
for $g\geq2$,
\beq\label{triFform4}
x^{2g-2}\F^{\{3\}}_g(x)
=\sum_{k=g-1}^{5g-5}\frac{r_{g,k}}{(3-q)^{k}}\,,
\eeq
where $r_{g,g-1},\dots,r_{g,5g-5}$ are rational numbers.
Noting that 
\beq\label{pxlogform3}
x^{2k} \p_x^{2k} (\log w)=\sum_{\ell=0}^{3k-1}\frac{H_{k,\ell}}{(3-q)^{k+\ell}}\, , \quad  k\ge1\,,\, H_{k,0},\dots,H_{k,3k-1}\in\QQ\,,
\eeq
and using Theorem~\ref{trithm1} and formula~\eqref{triweqn2},
we can achieve a new proof of~\eqref{triFform4}. 
Multiplying~\eqref{triFform4} by $dx/x^{j+1}$ and taking the residue at $x=0$, we obtain

\begin{align}
	n_{g}(3^{2j})
	=72^j (2j)! \sum_{m=0}^j\sum_{\ell=g-1}^{5g-5} \frac{(\ell-1)_m}{2^{m+\ell} m!} \binom{\frac{3j}2-1}{j-m} r_{g,\ell}\,.\label{triformula3}
\end{align}

Another direct consequence of Theorem~\ref{trithm1} is that 
for $g\geq2$, the genus $g$ free energy $\F^{\{3\}}_g(x)$ can be written in the form
\beq\label{triFform2}
\F^{\{3\}}_g(x)=\sum_{k=1}^{5g-5}\frac{b_{g,k}}{(1-108w)^{k}}
+\sum_{\ell=1}^{2g-2}\frac{b'_{g,k}}{w^{k}}\,,
\eeq
where $b_{g,1},\dots,b_{g,5g-5},b'_{g,1},\dots,b'_{g,2g-2}$ are rational numbers.

Finally, introduce $p=p(x)$ by
\beq\label{tripchange}
p=\frac{108w}{1-108w}\,.
\eeq
\begin{cor}\label{trithm3}
For $g=1$, 
\beq\label{triF1inq}
\F^{\{3\}}_1 =-\frac1{12}\log x+\frac1{24} \log \biggl(\frac{p+3}3\biggr)+\zeta'(-1)\,.
\eeq
For $g\geq 2$, $\F^{\{3\}}_g(x)$ admits an expression of the form
	\beq\label{triFform3}
	x^{2g-2}\F^{\{3\}}_g(x)=\frac{B_{2g}}{4g(g-1)}
	+\sum_{m=1}^{5g-5} c_{g,m} \, p^m\,,
	\eeq
where $c_{g,1}, \dots, c_{g,5g-5}$ are rational numbers; moreover, $c_{g,1}=\cdots=c_{g,2g-2}=0$.
\end{cor}
\begin{proof}
By using~\eqref{tridw} and~\eqref{tripchange}, one can prove 
\beq\label{trlogdxInp}
x^k \p_x^{k} (\log w) =-(k-1)!+\sum_{\ell=0}^{k-1}G_{k,\ell} p^{k+\ell}
\,,\quad k\ge1\,,~ G_{k,0},\dots,G_{k,k-1}\in \QQ\,.
\eeq
Then by using Theorem~\ref{trithm1} and~\eqref{triweqn2},
we obtain \eqref{triF1inq}, \eqref{triFform3}.
For $g\ge2$, recalling $x^{2g-2}\F^{\{3\}}_g-\frac{B_{2g}}{4g(g-1)}\in x^{2g-1}\QQ[[x]]$ and 
noticing $p\in x\QQ[[x]]$, we find $c_{g,1}=\cdots=c_{g,2g-2}=0$.
\end{proof}
Explicitly, 
\begin{align}
&\F^{\{3\}}_2 =\frac1{x^2}\biggl(-\frac{1}{240}+\frac{7 p^5}{12960}+\frac{29 p^4}{10368}+\frac{35 p^3}{10368}\,\biggr)\,, \nn \\
&\F^{\{3\}}_3  =\frac1{x^4}\biggl(\frac{1}{1008}+\frac{245 p^{10}}{1679616}+\frac{965 p^9}{559872}+\frac{2945 p^8}{373248}+\frac{813587 p^7}{47029248}+\frac{29969 p^6}{1679616}+\frac{5005 p^5}{746496}\biggr) \,.  \nn
\end{align}
Multiplying~\eqref{triFform3} by $dx/x^{j+1}$ and taking the residue at $x=0$, we obtain for $g\ge2$
\begin{align}
	n_{g}(3^{2j})=108^j (2j)!\sum_{\ell=2g-1}^{j} c_{g,\ell}\sum_{m=0}^{j-\ell}
	\frac{(\frac{3j}2-m)_m }{m!}\frac{(\frac{j}{2}+1)_{j-\ell-m}}{(j-\ell-m)!}(-3)^{-j+\ell+m}
	\label{triformula2}
\end{align}
($c_{g,\ell}$ are defined as~0 if $\ell>5g-5$).

\section{On quadrangulations and $(2\nu)$-angulations}\label{quadrangulation}
In this section we prove Theorem~\ref{quadthm1}.

Taking $\bs={\bf 1}_4$ in~\eqref{Ptype1} and \eqref{Ptype2}, we have
\begin{align}
&V^{\{4\}}(x,\e)\equiv 0 \,,\\
& 4 \, W^{\{4\}}(x,\e)(W^{\{4\}}(x,\e)+W^{\{4\}}(x+\e,\e)+W^{\{4\}}(x-\e,\e)) \,-\, W^{\{4\}}(x,\e) + x = 0 \label{differencequad}
\end{align}
(cf.~also~e.g.~\cite{BK, CL17, fik91}).
Substituting~\eqref{topoWbexpand} in~\eqref{differencequad}, we find $w=W^{\{4\}}_0(x)$ satisfies~\eqref{quadwgenus0eqn},
and for $g\geq1$, 
\beq\label{wgrec}
W^{\{4\}}_{g} = \frac4{1-24 w} \Biggl(\sum_{g_2=1}^{g-1} W^{\{4\}}_{g_2} W^{\{4\}}_{g-g_2}+2 \sum_{0\leq g_1,g_2\leq g-1, \, j\geq0 \atop g_1+g_2+j=g} \frac1{(2j)!} W^{\{4\}}_{g_2} \p_x^{2j}(W^{\{4\}}_{g_1})\Biggr)\,.
\eeq
From equation~\eqref{quadwgenus0eqn} we get 
\beq\label{quaddw}
\p_x=\frac1{1-24 w}\p_{w}\,.
\eeq
\begin{lemma}\label{quadDxlem}
For $k\geq 1$,
\beq\label{quaddxk}
\p_x^k = \sum_{i=1}^k \frac{\widetilde{T}_{k,i}}{(1-24w)^{2k-i}}\p_{w}^i\,,
\eeq
where $\widetilde{T}_{k,1},\dots,\widetilde{T}_{k,k}$ are rational numbers.
\end{lemma}
\begin{proof}
For $k=1$, formula~\eqref{quaddxk} is given by~\eqref{quaddw}. 
Suppose that~\eqref{quaddxk} is true for $k=m$ with $\widetilde{T}_{k,i}$ being rational numbers. 
Then for $k=m+1$, we have
\beq
\p_x^{m+1}=\p_x\circ\p_x^{m}=
\frac1{1-24 w}\p_{w} \circ \sum_{i=1}^{m}\frac{\widetilde{T}_{m,i}}{(1-24w)^{2m-i}}\p_{w}^i\nn\\
=\sum_{i=1}^{m+1}\frac{\widetilde{T}_{m+1,i}}{(1-24w)^{2m+2-i}}\p_{w}^i\,.
\eeq
Here
\beq
\widetilde{T}_{m,i}=24(2m-i)\widetilde{T}_{m,i}+\widetilde{T}_{m,i-1}\,,
\eeq
where $\widetilde{T}_{m,0}$ and $\widetilde{T}_{m,m+1}$ are defined as 0.
Obviously $\widetilde{T}_{m,i}$ are rational numbers.
By mathematical induction the lemma is proved.
\end{proof}
\begin{prop}\label{quadWprop}
For $g\geq 1$, $W^{\{4\}}_{g}(x)$ is given by
\beq\label{quadansatzOfSol}
W^{\{4\}}_{g}(x)=\frac{w\,P_g(w)}{(1-24w)^{5g-1}}\Bigl|_{w=w(x)}\Bigr.\,,
\eeq
for some polynomial $P_g(z)$. Moreover, $\deg P_g(w)\leq g-1$.
\end{prop}
\begin{proof}
For $g=1$, formula~\eqref{wgrec} gives
\beq
W^{\{4\}}_{1}(x)=4\frac{W^{\{4\}}_{0}(x)\,\p_x^2 (W^{\{4\}}_{0}(x))}{1-24W^{\{4\}}_{0}(x)}=\frac{96\,w}{(1-24 w)^4}\Big|_{w=w(x)}\,,
\eeq
thus~\eqref{quadansatzOfSol} holds. Suppose that~\eqref{quadansatzOfSol} holds for $g\leq m$
with $\deg P_g(w)\leq g-1$. Then for $g=m+1$, 
by using Lemma~\ref{quadDxlem} in the $g=m+1$ equation of~\eqref{wgrec}, we obtain
\beq
W^{\{4\}}_{m+1}(x)=\frac{w\,P_{m+1}(w)}{(1-24w)^{5m+4}}\Bigl|_{w=w(x)}\Bigr.\,,
\eeq
where $P_{m+1}(w)$ is a polynomial of $w$ with $\deg P_{m+1}(w)\leq g-1$.
By mathematical induction the proposition is proved.
\end{proof}

Proposition~\ref{quadWprop} says that
\beq\label{quadW}
W^{\{4\}}(x,\e)=w+\sum_{g\geq 1}\e^{2g}\sum_{\ell=4g-1}^{5g-1}\frac{\widetilde{A}_{g,\ell}}{(1-24w)^{\ell}}\,,
\eeq
for some rational numbers $\widetilde{A}_{g,4g-2}, \dots, \widetilde{A}_{g,5g-1}$.
Substituting~\eqref{quadW} in~\eqref{differencequad}, 
we find that 
$
U=2^{\frac{11}5}3^{\frac45}\sum_{g\geq 0}\widetilde{A}_{g,5g-1}(2^{\frac85}3^{\frac25}X)^{\frac{1-5g}{2}}
$
again satisfies the Painlev\'e I equation~\eqref{painleveI}.	
Therefore,
\beq
\widetilde{A}_{g,5g-1}
=\frac{48^g}{24}C_g\,, \quad g\geq 1\,,
\eeq
where $C_g$ is again the same universal constant as we found for triangulations~\eqref{univ_const}.

Taking $\bs={\bf 1}_4$ in~\eqref{defVWintro}, we have
\beq
\log W^{\{4\}}(x,\e) = \F^{\{4\}}(x+\e,\e) + \F^{\{4\}}(x-\e,\e) \,-\, 2 \, \F^{\{4\}}(x,\e) \,. \label{quadWF}
\eeq
Similar to~\eqref{wf}, we have
\begin{align}\label{quadwf}
&\F^{\{4\}}(x,\e) 
= \frac{1}{\e^2 \, \p_x^2} \, \sum_{k\geq 0} \frac{(1-2k) B_{2k}}{(2k)!} \, \e^{2k}\p_x^{2k} \bigl( \log W^{\{4\}}(x,\e) \bigr) \,.
\end{align}
\begin{proof}[Proof of Theorem~\ref{quadthm1}]
The $g=0$ case has been proved above (cf.~also~\eqref{F0dub}, \cite{DuY} or 
\cite[Proposition~3, Theorem~5]{DuY2}). 
For $g=1$, the statement can be verified directly using~\eqref{wgrec} (cf.~also~\eqref{fgfmgequal}--\eqref{jetF1} or~\eqref{Feven1general}).
For $g\ge2$, substituting~\eqref{topoWbexpand} in~\eqref{quadwf} and using~\eqref{quadansatzOfSol}, we obtain
\begin{align}\label{quadFgansatz1}
\F^{\{4\}}_g=& \frac{1-2g}{(2g)!}B_{2g}\p_x^{2g-2}(\log w)\nn \\
&+
\sum_{k=1}^{g-1} \frac{1-2k}{(2k)!}B_{2k}\p_x^{2k-2}
\biggl(\sum_{\ell=1}^{g-k}\frac{(-1)^{\ell-1}}{\ell (1-24w)^{5g-5k-\ell}}
\sum_{m_1,\dots,m_{\ell}\geq 1\atop m_1+\cdots+m_{\ell}=g-k}\prod_{i=1}^{\ell}P_{m_i}(w)\biggr)\nn\\
&+\p_x^{-2}
\biggl(\sum_{\ell=1}^{g}\frac{(-1)^{\ell-1}}{\ell (1-24w)^{5g-\ell}}
\sum_{m_1,\dots,m_{\ell}\geq 1\atop m_1+\cdots+m_{\ell}=g}\prod_{i=1}^{\ell}P_{m_i}(w)\biggr)\nn\\
=&\frac{1-2g}{(2g)!}B_{2g}\p_x^{2g-2}(\log w) +\sum_{\ell=4g-4}^{5g-5}\frac{
	\tilde{a}_{g,\ell}}{(1-24w)^{\ell}}
+\alpha_g w(1-12w)+\beta_g\,.
\end{align}
Here $\tilde{a}_{g,4g-4},\dots, \tilde{a}_{g,5g-5}$ are rational numbers and $\alpha_g, \beta_g$ are integration constants.
Using 
$
\p_x (\log w)=\frac{1}{(1-24 w) w}
$
and~\eqref{quaddw} we can prove by induction on $k$ that for $k\geq 1$,
\beq
\p_x^k (\log w) = O(w^{-2k})\,, \quad w \to \infty\,.
\eeq
Then 
from~\eqref{fgfmgequal} (cf.~\cite[Theorem 5]{DuY2}) we know that 
$\F^{\{4\}}_g \to 0$ as $w\to\infty$. Therefore, $\alpha_g=\beta_g=0$.
\end{proof}

Multiplying both sides of~\eqref{quadFform} by $dx/x^{3-2g+k}$ and taking the residue at $x=0$, we obtain
\begin{align}
n_{g}(4^k)&=\tilde{\sigma}_{g,k}+12^{k-2g+2}k! \sum_{m=0}^{k-2g+2}\binom{2k-4g+4-m}{k-2g+2-m}\Bigl(\sum_{\ell=4g-4}^{5g-5}\tilde{a}_{g,\ell}(\ell-1)_m\Bigr)\frac{2^m}{m!}\,,
\end{align}
where
$
\tilde{\sigma}_{g,k}\coloneqq k!\mathop{\rm res}\limits_{w=0}\frac{(1-24w)dw}{w^{k-2g+3}(1-12w)^{k-2g+3}}\frac{(1-2g)B_{2g}}{(2g)!}\p_x^{2g-2}(\log w)
$.

Another direct consequence of Theorem~\ref{quadthm1} is that 
	for $g\geq2$, $\F^{\{4\}}_g(x)$ admits the form
	\beq\label{quadFform2}
	\F^{\{4\}}_g(x)
	=\sum_{k=1}^{5g-5}\frac{\tilde{b}_{g,k}}{(1-24w)^k}
	+\sum_{k=1}^{2g-2}\frac{\tilde{b}'_{g,k}}{w^{k}}\,,
	\eeq
where $\tilde{b}_{g,1},\dots,\tilde{b}_{g,5g-5},\tilde{b}'_{g,1},\dots,\tilde{b}'_{g,2g-2}$ are rational numbers satisfying $\tilde{b}_{g,2}=\tilde{b}_{g,4}=\cdots=\tilde{b}_{g,4g-6}=0$.

Setting 
\beq\label{quadpchange}
p=\frac{24w}{1-24w}
\eeq
leads to the following corollary.
\begin{cor}\label{quadthm3}
For $g=1$, 
\beq
\F^{\{4\}}_1(x) =-\frac1{12}\log x+\frac1{12} \log \biggl(\frac{p+2}2\biggr)+\zeta'(-1)\,.
\eeq
	For $g\geq 2$, $\F^{\{4\}}_g(x)$ has the following expression:
	\beq\label{quadFform3}
	x^{2g-2}\F^{\{4\}}_g(x)
	=\frac{B_{2g}}{4g(g-1)}
	+\sum_{m=1}^{5g-5} \tilde{c}_{g,m} \, p^m,
	\eeq
	where $\tilde{c}_{g,1}, \dots, \tilde{c}_{g,5g-5}$ are rational numbers. Moreover,
	$\tilde{c}_{g,1}=\cdots=\tilde{c}_{g,2g-2}=0$.
\end{cor}
\begin{proof}
The proof based on Theorem~\ref{quadthm1} is similar to that of Corollary~\ref{trithm3}, so is omitted.
\end{proof}
Multiplying both sides of~\eqref{quadFform3} by $dx/x^{k+1}$ and taking the residue at $x=0$, we obtain
\begin{align}
n_{g}(4^{k})=24^k k!\sum_{\ell=2g-1}^{k} \tilde{c}_{g,\ell}\sum_{m=0}^{k-\ell}
\frac{(2k-m)_m }{m!}\frac{(k+1)_{k-\ell-m}}{(k-\ell-m)!}\,(-2)^{-k+\ell+m}
\end{align}
($\tilde{c}_{g,\ell}$ are defined as~0 if $\ell>5g-5$).

Following~\cite{E}, introduce $q$ by
\beq\label{quadqchange}
w=\frac{q-1}{12q}\,.
\eeq
\begin{cor}\label{quadthm4}
For $g=1$,
\beq
\F^{\{4\}}_1(x) =-\frac1{12}\log x-\frac1{12}\log (2-q)+\zeta'(-1).
\eeq
	For $g\geq2$, $\F^{\{4\}}_g(x)$ admits an expression of the form
	\beq\label{quadFform4}
	x^{2g-2}\F^{\{4\}}_g(x)
	=\sum_{k=2g-2}^{5g-5}\frac{\tilde{r}_{g,k}}{(2-q)^{k}}\,,
	\eeq
	where $\tilde{r}_{g,2g-2},\dots,\tilde{r}_{g,5g-5}$ are rational numbers.
\end{cor}
The proof is again omitted.

It was proved in~\cite{E2} that, for any fixed $\nu\geq 2$, the generating function
$
e_g(s)\coloneqq\sum_{k\geq 1}\frac{n_g((2\nu)^k)}{k!}s^k
$
admits the expression
\beq\label{Eregexpression}
e_g(s)=C^{(g)}+\sum_{\ell=2g-2}^{5g-5}\frac{\tilde{r}_{g,\ell}(\nu)}{(\nu-(\nu-1)z_0)^\ell}\,, \quad g\geq 2\,.
\eeq
Here, $\tilde{r}_{g,2g-2}(\nu),\dots,\tilde{r}_{g,5g-5}(\nu)$ are some rational numbers, $C^{(g)}$ is a certain constant 
independent of~$\nu$, and $z_0=z_0(s)=1+\frac{(2\nu)!}{\nu!(\nu-1)!}s+\cdots$ is the unique solution to 
\beq\label{z0eqn}
z_0=1+\frac{(2\nu)!}{\nu!(\nu-1)!}s\, z_0^{\nu}\,,
\eeq
Here, 
$z_0(x^{\nu-1})=q(x)$.
For the case when $\nu=2$, Eynard~\cite{Eynard} gave an expression of $e_g(s)$, which is equivalent to~\eqref{Eregexpression}. 
Ercolani--Lega--Tippings~\cite[Conjecture~5.2]{ELT} conjectured that
\beq\label{ELTconj}
C^{(g)}=-\frac{B_{2g}}{4g(g-1)},\quad g\geq 2.
\eeq
By using~\eqref{Fgexplicit}, \eqref{qeqn}, \eqref{z0eqn} and based on~\eqref{Eregexpression}, we find that Ercolani--Lega--Tippings's conjecture~\eqref{ELTconj}
is equivalent to Conjecture~A.

\begin{proof}[Proof of Theorem~A]
By comparing~\eqref{evenFform4} with~\eqref{quadFform4}.
\end{proof}
We note that the topological recursion (see e.g.~\cite{Eynard}) should also lead to a proof of Theorem~A.

Now we use a formula from~\cite{E2} to give a second proof of Theorem~A. It was already shown 
in~\cite{E2} that the $C^{(g)}$ for any $\nu$ satisfy the following recursion:
\beq\label{ErcolaniRecursive}
\frac{(2g-1)!}{(2g+2)!}-\frac{(2g-1)!}{12(2g)!}
+\sum_{k=2}^{g}\frac{(1-2g)_{2g-2k+2}}{(2g-2k+2)!}C^{(k)}=0\,, \quad g\geq 2\,.
\eeq
\begin{proof}[A second proof of Theorem~A]
Define $\widetilde{C}^{(0)}\coloneqq1$, $\widetilde{C}^{(1)}\coloneqq1/6$, and 
$\widetilde{C}^{(g)}\coloneqq-4g(g-1)C^{(g)}$ for $g\geq 2$.
Then it follows from~\eqref{ErcolaniRecursive} that
\beq\label{Recursive}
\sum_{k=0}^g \frac1{(2g-2k+2)!} \frac{2k-1}{(2k)!}\widetilde{C}^{(k)}
=-\frac12 \delta_{g,0}\,, \quad g\geq 0\,.
\eeq
For a formal variable $y$, mutiplying~\eqref{Recursive} by $y^g$ and summing over $g$, we find
\beq\label{generatingSeries}
\biggl(\sum_{\ell\geq 0}\frac{y^{2\ell+2}}{(2\ell+2)!}\biggr)
\biggl(\sum_{k\geq 0}\frac{2k-1}{(2k)!}\widetilde{C}^{(k)}y^{2k-2}\biggr)=-\frac12\,,
\eeq 
thus
\beq\label{seriesStep1}
\sum_{k\geq 0}\frac{2k-1}{(2k)!}\widetilde{C}^{(k)}y^{2k-2}=-(e^{\frac{y}2}-e^{-\frac{y}2})^{-2}\,.
\eeq
Integrating~\eqref{seriesStep1} with respect to~$y$ and multiplying both sides by $y$, we find
\begin{align}
\sum_{k\geq 0}\frac{\widetilde{C}^{(k)}}{(2k)!}y^{2k}=\frac{y}2\frac{e^y+1}{e^y-1}
=\frac12\Bigl(\frac{y}{e^y-1}+\frac{-y}{e^{-y}-1}\Bigl)=\sum_{k\geq 0}\frac{B_{2k}}{(2k)!}y^{2k}\,.
\end{align}
Hence $\widetilde{C}^{(g)}=B_{2g}$. 
Theorem~A is proved.
\end{proof}

\section{More on the $b=2\nu$ case}\label{evenangulation}
In this section, we prove Theorem~\ref{rglthm} and give some more discussions on the $b=2\nu$ case.

As a particular example of equation~\eqref{evengenus0eqn},
we know that $w$ satisfies
\beq\label{2nugenus0eqn}
x=w-\frac{(2\nu)!}{\nu!(\nu-1)!}w^{\nu}\,.
\eeq
Obviously the variable $q$ introduced in~\eqref{qeqn} is related to~$w$ by
$q=w/x$.
\begin{lemma}\label{evendxlogw}
For $k\geq1$, we have
\beq\label{xdxIny}
\p_x^k (\log w)
=\frac1{x^k}\sum_{\ell=0}^{k-1}\frac{(-1)^{k-\ell+1}e_{k,\ell}(\nu)}{(\nu-(\nu-1)q)^{k+\ell}}\,,
\eeq
where $e_{k,\ell}(\nu)\in \ZZ[\nu]$, $\ell=0,\dots,k-1$, with $e_{k,k-1}(\nu)=(2k-3)!!\nu^{k-1}$ and
$\deg e_{k,\ell}(\nu)\leq k-1$.
\end{lemma}
\begin{proof}
For $k=1$, the statement can be proved easily 
(here we recall the convention that $(-1)!!=1$).
Suppose that the statement is true for $k=m$. Then for $k=m+1$, we have
\beq
\p_x^{m+1} (\log w)
=\frac1{x^{m+1}}\sum_{\ell=0}^{m}\frac{(-1)^{m-\ell+2}e_{m+1,\ell}(\nu)}{(\nu-(\nu-1)q)^{m+1+\ell}}\,,
\eeq
where 
\beq\label{evenrecursiveE}
e_{m+1,\ell}(\nu)=(\ell+1)e_{m,\ell+1}(\nu)+(m+\ell)(\nu+1)e_{m,\ell}(\nu)+(m+\ell-1)\nu \,e_{m,\ell-1}(\nu)\,,
\eeq 
with $e_{m,-2}(\nu)=e_{m,-1}(\nu)=e_{m,m}(\nu)=e_{m,m+1}(\nu)=0$.
Therefore, $e_{m+1,\ell}(\nu)\in\ZZ[\nu]$, $\deg e_{m,\ell}(\nu)\leq m$
 and $e_{m+1,m}=(2m-1)!!\nu^m$. 
The lemma is proved.
\end{proof}
Using~\eqref{fgfmgequal}, \eqref{jetF1} and Lemma~\ref{evendxlogw}, we obtain
\beq
\F_1^{\{2\nu\}}(x)=-\frac1{12}\log x-\frac1{12}\log (\nu-(\nu-1)q)+\zeta'(-1)\,.
\eeq
\begin{proof}[A third proof of Theorem~A]
According to the Hodge--GUE correspondence \cite{DLYZ20,DuY2}, the function $\F^{\{2\nu\}}_g(x)$ is a linear combination of
$\prod_{i=1}^{\ell(\lambda)}\p_x^{\lambda_i}(\log w)/(\p_x(\log w))^{\ell+2g-m-1-k}$
with rational coefficients,
where $1\leq m\leq g$, $0\leq k\leq 3m-3$, $\ell\geq 1$, and $\lambda$ are partitions satisfying $\ell(\lambda)=2g-2m+\ell$, $|\lambda|=4g-m-3-k+\ell$. 
	From~\eqref{xdxIny} we know that
	\beq\label{vlambda}
	\frac{\prod_{i=1}^{\ell(\lambda)}\p_x^{\lambda_i}(\log w)}{(\p_x(\log w))^{\ell+2g-m-1-k}}=\frac{(\nu-(\nu-1)q)^{2-2g}}{x^{2g-2}}
	\sum_{0\leq s_i\leq\lambda_i-1}
	\frac{\prod_{i=1}^{\ell(\lambda)}(-1)^{\lambda_i-s_i+1}e_{\lambda_i,s_i}(\nu)}{(\nu-(\nu-1)q)^{|s|}}\,.
	\eeq
Formula~\eqref{evenFform4} then follows from the fact that $-(5g-5)\leq 2-2g-|s|\leq -(2g-2)$.
\end{proof}

\begin{proof}[Proof of Theorem~\ref{rglthm}]
The fact that $r_{g,\ell}(\nu)\in \QQ[\nu]$ follows from~\eqref{vlambda} and that $e_{g,\ell}(\nu)\in \ZZ[\nu]$,
which we proved in Lemma~\ref{evendxlogw}.
Since $\deg\,e_{k,\ell}(\nu)\leq k-1$, we have
\beq
\deg\bigl(e_{\lambda_1,s_1}(\nu)\cdots e_{\lambda_{\ell(\lambda)},s_{\ell(\lambda)}}(\nu)\bigr)
=|\lambda|-\ell(\lambda)=2g+m-3-k\leq 3g-3\,,
\eeq
for $0 \leq s_i \leq\lambda_i$, $i=1,\dots,\ell(\lambda)$. 
Hence $\deg\,r_{g,\ell}(\nu)\leq 3g-3$.
\end{proof}

Similar to~\cite{DYZ}, by using Theorem~A we can prove that, for $2m-2+j>0$,
\beq\label{evendxFm}
\p_x^j (\F^{\{2\nu\}}_m(x))=\frac1{x^{2m-2+j}}\sum_{\ell=2m-2+j}^{5m-5+2j}\frac{r_{m,\ell,j}(\nu)}{(\nu-(\nu-1)q)^{\ell}}\,.
\eeq
By~\eqref{defVWintro} we have
\beq\label{evenWstep1}
W^{\{2\nu\}}_{g}(x)=w\sum_{n= 1}^{g}\frac{2^n}{n!}\sum_{k_1+m_1,\dots,k_n+m_n\geq 1 \atop |k|+|m|=g}
\prod_{i=1}^{n}\frac{\p_x^{2k_i+2}(\F_{m_i}^{\{2\nu\}}(x))}{(2k_i+2)!}\,,\quad g\geq 1\,.
\eeq 
Using~\eqref{evendxFm} and \eqref{evenWstep1}, we obtain 
\beq\label{evenWg}
x^{2g-1} W^{\{2\nu\}}_{g}(x)=q\sum_{\ell=2g}^{5g-1}\frac{t_{g,\ell}(\nu)}{(\nu-(\nu-1)q)^{\ell}}\,, \quad g\ge1\,.
\eeq
Using Theorem~\ref{rglthm}, we  show that $r_{m,\ell,j}(\nu)$ are polynomials in $\nu$ of degree at most $3m-3+j$.
Hence, we obtain:

\begin{cor}\label{propwgconj}
For any $g\ge1$, $t_{g,\ell}(\nu)$ are polynomials of~$\nu$ with $\deg\, t_{g,\ell}(\nu)\leq 3g-1$.
\end{cor}

We note that part of the motivation of Corollary~\ref{propwgconj} comes from~\cite{GL}. Indeed, following~\cite{GL}, define
\beq
Q_{g,k}(\nu) := k![x^{1-2g+(\nu-1)k}]W^{\{2\nu\}}_{g}(x)/\Bigl(\frac{(2\nu)!}{\nu!(\nu-1)!}\Bigr)^k\,.
\eeq
(As shown in~\cite{EMP}, 
$k![x^{1-2g+(\nu-1)k}]W^{\{2\nu\}}_{g}(x)$ counts the number of two-legged $2\nu$-valent genus-$g$ maps with $k$ vertices.)
It then follows from~\eqref{qeqn} and~\eqref{evenWg} that
\beq\label{Qgk}
Q_{g,k}(\nu)=k! \sum_{m=0}^k \binom{\nu k}{k-m}\frac{(\nu-1)^m}{m!}\sum_{\ell=2g}^{5g-1} t_{g,\ell}(\nu)
(\ell-1)_m\,,
\eeq
which, together with Corollary~\ref{propwgconj}, implies  
validity of a conjectural statement in~\cite[Remark~2.9]{GL} that 
$Q_{g,k}(\nu)$ are polynomials in~$\nu$ of degree $3g-1+k$.

Finally, we note that it can be deduced using Theorem~A and a result of~\cite{E2} that
\beq\label{rgCg}
r_{g,5g-5}(\nu)=\frac{\nu^{3g-3}}{12^{g}}\frac{C_g}{(5g-3)(5g-5)}\,, \quad g\ge2\,,
\eeq 
which can also be proved 
using the Hodge--GUE correspondence (cf.~\cite{DLYZ20} and~\cite[Corollary~4.6]{Y2}) and~\cite[(6.10)]{IZ}. 
Here $C_g$ are defined by~\eqref{painleveI}.
By \eqref{evendxFm}--\eqref{evenWg} and~\eqref{rgCg} we have 
\beq
t_{g,5g-1}(\nu)=\frac{\nu^{3g-1}}{12^g}C_g\,.
\eeq


\begin{thebibliography}{99}

\bibitem{AvM}
Adler, M., van Moerbeke, P., Matrix integrals, Toda symmetries, Virasoro constraints, and orthogonal polynomials. Duke Math. J.,~{\bf 80} (1995), 863--911.

\bibitem{ACKM}
Ambj{\o}rn, J., Chekhov, L., Kristjansen, C.F., Makeenko, Yu., 
Matrix model calculations beyond the spherical limit. Nucl. Phys.~B,~{\bf 404} (1993), 127--172.

\bibitem{BCGE}
Belliard, R., Charbonnier, S., Garcia-Failde, E., Eynard, B., 
Topological recursion for generalised Kontsevich graphs and $r$-spin intersection numbers. Sel.~Math.~New Ser.,~{\bf 31} (2025), Paper No.~88.

\bibitem{BIZ}
Bessis, D., Itzykson, C., Zuber, J. B., Quantum field theory techniques in graphical enumeration. Adv. Appl. Math., {\bf 1} (1980), 109--157.

\bibitem{BD}
Bleher, P.M., Dea\~no, A., 
Topological expansion in the cubic random matrix model.
 IMRN {\bf 2013} (2013), 2699--2755.

\bibitem{BG17}
Borot, G., Garcia-Failde, E., 
Simple Maps, Hurwitz Numbers, and Topological Recursion.
Commun. Math. Phys.,~{\bf 380} (2020), 581--654.

\bibitem{BIPZ}
Br\'ezin, E., Itzykson, C., Parisi, P., Zuber, J. B., Planar diagrams. Commun. Math. Phys.,~{\bf 59} (1978), 35--51.

\bibitem{BK}
Br\'ezin, E., Kazakov, V.,
Exactly solvable field theories of closed strings.
Phys. Lett. B,~{\bf 236} (1990), 144--150.

\bibitem{CL17}
Chan, C.-T., Liu, H.-F.,
Graphic enumerations and discrete Painlev\'e equations via random matrix models.
arXiv:1712.09231.

\bibitem{CMS}
Chapuy, G., Marcus, M., Schaeffer, G.,
A Bijection for Rooted Maps on Orientable Surfaces. SIAM J.~Discrete Math., {\bf 23} (2009), 1587--1611.

\bibitem{DW}
Dijkgraaf, R., Witten, E.,
Mean field theory, topological field theory, and multi-matrix models.
Nucl. Phys. B,~{\bf 342} (1990), 486--522.

\bibitem{DM}
Do, N., Manescu, D.,
Quantum curves for the enumeration of ribbon graphs and hypermaps. 
Commun. Number Theory Phys.,~{\bf 8} (2014), 677--701.


\bibitem{DS}
Douglas, M., Shenker, S.,
Strings in less than one dimension.
Nucl.~Phys. B,~{\bf 335} (1990), 635--654.

\bibitem{Du96}
Dubrovin, B., Geometry of 2D topological field theories. In: Francaviglia, M., Greco, S. (eds.) ``Integrable Systems and Quantum Groups" (Montecatini Terme, 1993), Lecture Notes in Math., vol.~{\bf 1620}, pp. 120--348. Springer, Berlin, 1996.

\bibitem{Du2}
Dubrovin, B., Hamiltonian perturbations of hyperbolic PDEs: from classification results to the properties of solutions. 
In: Sidoravicius, V. (ed.) ``New Trends in in Mathematical Physics". Selected Contributions of the XVth International 
Congress on Mathematical Physics, pp. 231--276. Springer, Dordrecht, 2009.

\bibitem{DLYZ20}
Dubrovin, B., Liu, S.-Q., Yang, D., Zhang, Y.,
Hodge--GUE Correspondence and the Discrete KdV Equation.
Commun. Math. Phys.,~{\bf 379} (2020), 461--490.

\bibitem{DuY}
Dubrovin, B., Yang, D., Generating series for GUE correlators. 
Lett. Math. Phys.,~{\bf 107} (2017), 1971--2012.

\bibitem{DuY2}
Dubrovin, B., Yang, D., On cubic Hodge integrals and random matrices.
Commun. Number Theory Phys.,~{\bf 11} (2017), 311--336.

\bibitem{DYZ}
Dubrovin, B., Yang, D., Zagier, D., 
Classical Hurwitz numbers and related combinatorics.
Mosc. Math. J.,~{\bf 17} (2017), 601--633. 

\bibitem{DuZ}
Dubrovin, B., Zhang, Y., Virasoro symmetries of the extended Toda hierarchy. Commun. Math. Phys.,~{\bf 250} (2004), 161--193.

\bibitem{E}
Ercolani, N.M., Caustics, counting maps and semi-classical asymptotics. Nonlinearity,~{\bf 24} (2011), 481--526.

\bibitem{E2}
Ercolani, N.M., Conservation laws of random matrix theory. In: Deift, P., Forrester, P. (eds.)
``Random matrix theory, interacting particle systems, and integrable systems".
Math. Sci. Res. Inst. Publ.,~{\bf 65}, Cambridge University Press, New York, 2014.

\bibitem{ELT2}
Ercolani, N.M., Lega, J., Tippings, B.,
Non-recursive counts of graphs on surfaces.
Enumer. Comb. Appl.,~{\bf 3} (2023), Paper No. S2R20, 24~pp. 

\bibitem{ELT}
Ercolani, N.M., Lega, J., Tippings, B., 
Map enumeration from a dynamical perspective. 
In ``Recent progress in special functions", Contemp. Math.,~{\bf 807}, Amer. Math. Soc., RI, 2024, pp. 85--110. 
arXiv:2308.06369v3.

\bibitem{EMP}
Ercolani, N.M., McLaughlin, K.D.T.-R., Pierce, V.-U.,
Random matrices, graphical enumeration and the continuum limit of toda lattices.
Commun. Math. Phys.,~{\bf 278} (2008), 31--81.

\bibitem{Eynard}
Eynard, B., Counting surfaces. vol.~{\bf 70}. Prog. Math. Phys., Birkh{\"a}user Basel, 414~pp, 2016.

\bibitem{EO09}
Eynard, B., Orantin, N., Topological recursion in enumerative geometry and random matrices. 
J.~Phys.~A: Math.~Theor.,~{\bf 42} (2009), Paper No.~293001.

\bibitem{Flaschka}
Flaschka, H., On the Toda lattice. II. Inverse-scattering solution. Progr. Theoret. Phys.,~{\bf 51} (1974), 703--716.

\bibitem{fik91}
Fokas, A., Its, A., Kitaev, A.,
Discrete Painlev\'e equations and their appearance in quantum gravity.
Commun. Math. Phys.,~{\bf 142} (1991), 313--344.

\bibitem{GMMMO}
Gerasimov, A., Marshakov, A., Mironov, A., Morozov, A., Orlov, A.,
Matrix models of two-dimensional gravity and Toda theory. 
Nucl. Phys. B,~{\bf 357} (1991), 565--618.

\bibitem{GL}
Gharakhloo, R., Latimer, T.,
Combinatorics of even-valent graphs on Riemann surfaces.
arXiv:2505.01633.

\bibitem{GMM}
Giacchetto, A.,  Maity, P., Mazenc, E.A.,
Matrix Correlators as Discrete Volumes of Moduli Space I: Recursion Relations, the BMN-limit and DSSYK.
arXiv:2510.17728.

\bibitem{GM}
Gross, D., Migdal A.,
Non-perturbative two-dimensional quantum gravity.
Phys. Rev. Lett.,~{\bf 64}, (1990), 127--130.

\bibitem{HZ}
Harer, J., Zagier, D., The Euler characteristic of the moduli space of curves. Invent. Math.,~{\bf 85} (1986), 457--485.

\bibitem{IZ}
Itzykson, C., Zuber, J. B., Combinatorics of the modular group. II. The Kontsevich integrals,
Internat. J. Modern Phys. A,~{\bf 7} (1992), 5661--5705.

\bibitem{JK}
Joshi, N., Kitaev, A.V.: 
On Boutroux’s tritronqu\'ee solutions of the first Painlev\'e equation. Stud.
Appl. Math.,~{\bf 107} (2001), 253--291.

\bibitem{Ka}
Kapaev, A.A.: Quasi-linear Stokes phenomenon for the Painlev\'e first equation. 
J. Phys. A,~{\bf 37}
(2004), 11149--11167.

\bibitem{KLS}
Kramer, R., Lewa\'nski, D., Shadrin, S., 
Quasi-polynomiality of monotone orbifold Hurwitz numbers and Grothendieck’s dessins d'enfants.
Doc. Math.,~{\bf 24} (2019), 857--898.

\bibitem{MMMM}
Makeenko, Y., Marshakov, A., Mironov, A., Morozov, A.,
Continuum versus discrete Virasoro in one-matrix models. 
Nucl. Phys. B,~{\bf 356} (1991), 574--628. 

\bibitem{Manakov}
Manakov, S. V., Complete integrability and stochastization of discrete dynamical systems.
J. Experiment. Theoret. Phys.,~{\bf 67}, 543--555 (in Russian). English translation in: Soviet Physics JETP,~{\bf 40} (1974), 269--274.

\bibitem{Mehta}
Mehta, M.L.,
Random Matrices, 2nd edn.
Academic Press, Cambridge, 1991.

\bibitem{Mo}
Morozov, A.,
Integrability and matrix models.
Phys.-Uspekhi,~{\bf 37} (1994), 1--55.

\bibitem{Nor}
Norbury, P., Counting lattice points in the moduli space of curves. 
Math. Res. Lett.,~{\bf 17} (2010), 467--481.

\bibitem{T}
Tutte, W.T., 
On the enumeration of planar maps.
Bull. Amer. Math. Soc.,~{\bf 74} (1968), 64--74.

\bibitem{W}
Witten, E., Two-dimensional gravity and intersection theory on moduli space.
Surv. Differ. Geom.,~{\bf 1} (1991), 243--310.

\bibitem{Y}
Yang, D., On tau-functions for the Toda lattice hierarchy. Lett. Math. Phys.,~{\bf 110} (2020), 555--583.

\bibitem{Y2}
Yang, D., GUE via Frobenius Manifolds. I. From matrix gravity to topological gravity and back. Acta Math. Sin.,~{\bf 40} (2024), 383--405.

\bibitem{YZZ}
Yang, D., Zagier, D., Zhang, Y., Masur-Veech volumes of quadratic differentials and their asymptotics.
J. Geom. Phys.,~{\bf 158} (2020), Paper No.~103870, 12 pp.

\bibitem{YZ}
Yang, D., Zhou, J., Grothendieck's dessins d'enfants in a web of dualities. III. 
J. Phys. A,~{\bf 56} (2023), Paper No.~055201, 34 pp.
\end{thebibliography}
\end{document}